\documentclass[journal]{IEEEtran}
\usepackage{cite}
\usepackage{amsfonts,amssymb,amsthm,amsbsy,amsmath,paralist,bm,ifthen,color}
\usepackage{algorithm,algorithmic}
\usepackage{graphicx}
\usepackage{subfig}
\usepackage[font=footnotesize]{caption}
\usepackage{xcolor}
\usepackage{relsize}
\usepackage[numbers,sort&compress,square]{natbib}
\usepackage{tabularx}
\usepackage{booktabs}
\usepackage{multirow}
\usepackage{setspace,cases}
\newtheorem{lemma}{Lemma}
\newtheorem{theorem}{Theorem}

\newtheorem{remark}{Remark}

\ifCLASSINFOpdf

\else

\fi

\begin{document}
\title{Relay Selection and Resource Allocation for Ultra-Reliable Uplink Transmission in Smart Factory Scenarios}
\author{Jing Cheng,~\IEEEmembership{Student Member,~IEEE}, Chao Shen,~\IEEEmembership{Member,~IEEE}
\thanks{Jing Cheng and Chao Shen are with the State Key Laboratory of Rail Traffic Control and Safety, Beijing Jiaotong University, Beijing 100044, China (e-mail: chengjing@bjtu.edu.cn; chaoshen@bjtu.edu.cn).}
}

%
\maketitle
\begin{abstract}
  In this paper, a relay-aided two-phase transmission protocol for the smart factory scenario is proposed. This protocol aims at enabling all robots' ultra-reliable target number of uplink critical data transmission within a latency constraint by jointly optimizing the relay selection, resource block (RB) assignment, and transmit power allocation. Such protocol design is formulated as a mixed-integer and strictly non-convex problem where optimization variables are mutual coupling, which is definitely challenging. Instead of conventional methods designed for solving the problem, we leverage the properties of the relative entropy function to equivalently transform the problem without introducing extra constraints. As the packet error probability requirements of each robot under two possible transmission modes are coupled in one overall reliability constraint, the big-M technique is applied to decouple it into two corresponding reliability constraints. One is for direct transmission mode, and the other is for cooperative transmission mode. Moreover, both non-convex penalty (NCP) and quadratic penalty (QP) approaches are utilized to deal with the binary indicator constraints. Based on such penalty methods, a sequence of penalized approximated convex problems can be iteratively solved for sub-optimal solutions. Numerical results demonstrate the efficiency of such two penalty methods from the perspectives of sub-optimal values of total transmit power and convergence rate. Further, the impacts of reliability, the number and location of relays, the number of robots, the target number of data bits on the total power consumption are analyzed.
\end{abstract}

\IEEEpeerreviewmaketitle
\section{Introduction}
There is a paradigm shift of the fifth-generation ($5$G) network from the conventional human-centric communication to the Internet of things (IoT) machine-to-machine (M2M) communication \cite{HampelICM2019}. In the context of $5$G pillar use cases, IoT applications have been divided into two categories: massive IoT applications and mission-critical IoT applications. Specifically, massive IoT applications aim at improved connectivity, scalability, and battery life of a large number of low-power and low-cost devices \cite{SchulzICM2017}, which has been widely studied in the literature \cite{Chai2018,SenelITC2018,CuiIJoSAiC2021}. On the other hand, mission-critical IoT applications such as factory automation, remote surgery, as well as intelligent transport systems (ITS) feature ultra-reliable and low-latency communication (URLLC) \cite{ShafiIJSAC2017,ParvezICST2018}, which is of great challenge. Looking from the fourth industrial revolution roadmap (known as Industry $4.0$), the communication network in a smart factory is expected to shift from wired connections to wireless connections, thereby imposing stringent requirements of end-to-end delay (e.g. $1$ ms) and reliability (e.g. $1-10^{-9}$) on the wireless transmission \cite{AijazPI2019,HuangIIEM2018}.

It is the conflicting nature of low-latency and ultra-high reliability that makes mission-critical IoT applications challenging. On one hand, ultra-reliability can be achieved by using strong channel codes with low code rates, adding redundancy and parity, and leveraging retransmission techniques like hybrid automatic repeat request with chase combining (HARQ-CC) and HARQ with incremental redundancy (HARQ-IR) \cite{ShirvanimoghaddamICM2019,ZhangITC2020,CohenIToC2021,XuIToC2020}. Inevitably, such operations will result in increased latency. On the other hand, small-size packets can be transmitted to reduce latency, which can, in turn, cause a loss of coding gain \cite{JiIWC2018}. In fact, ultra-high reliability and ultra-low latency cannot be achieved simultaneously in the current long-packet wireless transmission system. Therefore, it is necessary to design an ultra-reliable transmission scheme suitable for wireless transmission scenarios in smart factories under a given specific delay requirement.

Recently, extensive research attention has been paid to ultra-reliable communication. Particularly, retransmission is one of the potential and promising technologies for enhancing reliability. In practice, the retransmission technique is widely used in $4$G Long-Term Evolution (LTE) to improve reliability at the cost of increased latency. Several attempts of applying retransmission techniques to URLLC-oriented scenarios have been made. For example, \cite{Kotaba2019,KotabaIToWC2021} investigated the uplink transmission performance by incorporating non-orthogonal multiple access (NOMA) with HARQ protocol while meeting high reliability requirements specified for all users. The authors in \cite{MakkiIWCL2014} considered the power-limited throughput of a HARQ-IR protocol with short-blocklength codes as a performance metric, where the closed-form expressions of the outage probabilities were derived. From the point of view of exploiting time diversity for enabling URLLC in an energy-efficient manner, the proposed HARQ-IR joint optimization problem over the number of retransmissions, blocklength, and power per round in \cite{AvranasIJSAC2018} was studied and solved by the dynamic programming algorithm. For a burst communication scenario, a delay-constrained fast HARQ protocol subject to the derived error probabilities with closed-form expressions was presented in \cite{MakkiITWC2019}. Admittedly, the retransmission mechanism can improve transmission reliability by exploiting time diversity, but the delay will inevitably increase.

On the other hand, URLLC packets scheduled in mini-slots experience the quasi-static fading channel, and such quasi-static channel characteristics will degrade the performance of retransmission. Moreover, when the channel is in deep fading for a long time, the time diversity may not be able to guarantee the reliability of URLLC packet transmission. Thus, we consider a two-phase relay-aided transmission protocol, which can act as an enabler for ultra-reliable communication by means of spatial diversity, apart from the advantages of improving capacity and quality of service (QoS). Traditionally, it is long packets that are transmitted in the relay-assisted wireless communication system, and the achievable rate is normally characterized by Shannon's capacity \cite{AlamIToWC2013,JuIToC2019,BedeerIJoSAiC2015}. However, due to the demanding latency requirement in mission-critical applications, the size of transmitted URLLC packets is small. Thus, the communication is no longer arbitrarily reliable and the decoding error probability is no longer negligible. In view of this, Shannon's capacity is not applicable to characterize the maximum achievable rate of short URLLC packets. Otherwise, the performance of the latency and reliability will be underestimated \cite{XuIToWC2020,XuIToVT2020}. This necessitates the achievable rate characterization and relaying protocol design under the finite blocklength regime. Initially, a tight normal approximation of small packets' achievable rate was derived in \cite{PolyanskiyITIT2010} for the AWGN channel. For further extensions, achievable rates of Gilbert-Elliott channels \cite{PolyanskiyIToIT2011}, quasi-static fading channels \cite{YangITIT2014,Gursoy2011}, quasi-static fading channels with retransmissions \cite{WuIToC2011,MakkiIWCL2014} and frequency-selective fading channels \cite{SheITWC2018} were subsequently characterized. The above fundamental results for short packets' approximated achievable rates under different channel conditions laid a foundation for the relaying protocol design with finite blocklength codes. For example, the blocklength-limited capacity performance of a simple relaying scenario with a source, a destination, and a relay was studied in \cite{HuIToWC2016,HuIToVT2016}. In view of the advantage of applying relaying to $5$G URLLC transmissions stated in \cite{HuIN2018}, different relay-assisted transmission schemes for URLLC were investigated. For critical data transmission in a vehicle-to-everything (V2X) network, a two-phase (broadcast phase and relaying phase) transmission policy was proposed \cite{GhatakIWCL2021}, where the outage probability was minimized by the cooperative frame design. In \cite{Singh2020}, an energy-efficient resource allocation strategy was proposed for the multi-user multicarrier amplify-and-forward (AF) relay network where $K$ URLLC paired users communicating via an AF relay in the multiple access phase and broadcast phase. Aside from studies of the half-duplex relay, \cite{Hu2019} presented the reliability-optimal design of a full-duplex decode-and-forward (DF) relaying network and showed that the full-duplex DF relaying generally outperforms the full-duplex AF relaying in terms of reliability, especially for mission-critical applications. Since the capacity formula in the finite blocklength regime itself is quite complex, the studies of relay-assisted URLLC transmission generally consider the simple scenario of a single relay with one/multiple users as the above works did.

Different from the existing works, in this paper, we consider the multiple robots' uplink critical data transmission in a smart factory scenario where the transmission reliability is enabled by the aid of multiple relays. Given that robots and relays are battery-limited, the design principle is set as total transmit power minimization under the predefined maximum overall packet error probability and delay constraints to prolong the lifetime. Meanwhile, in view that the target payload in critical IoT applications is generally small, the strategy of assigning multiple resource blocks (RBs) to one user is not efficient and the assigned RBs cannot be fully leveraged. Thus, in terms of RB assignment, the design policy is that only one RB can be occupied by one robot. In order to make each robot complete the critical data upload mission within target reliability and latency constraints in an energy-efficient manner, we propose a relay-aided two-phase transmission protocol that jointly optimizes the relay selection, RB assignment, and transmit power allocation. The main contributions of this paper are summarized as follows.
\begin{itemize}
  \item The problem on how to perform the relay selection, RB assignment and transmit power allocation to guarantees all robots' uplink critical data transmission under a predefined delay constraint in an energy-efficient manner is studied. The formulated problem is strictly non-convex in terms of the binary nature of relay selection and RB assignment indicators, mutual coupling of optimization variables, and neither convex nor concave rate expressions. By the proposed low-complexity and fast-convergence non-convex penalty (NCP) based and quadratic penalty (QP) based successive convex approximation (SCA) iterative algorithms, the stationary-point solutions of the tightly approximated formulated problem can be obtained.
  \item It is common in the literature that solving the coupling problem of optimization variables is to introduce a new auxiliary variable and then add additional constraints on the relationship between the new variable and the original one \cite{SunIToC2017}. Instead, in this paper, we fully excavate the properties of the relative entropy function to transform the original problem, and the transformed problem is equivalent to that dealt with by the traditional method. As compared to the traditional method, the proposed transformation has an advantage in the sense that no extra constraints are introduced.
  \item For each robot, the reliability requirements under two possible transmission modes are coupled in one constraint. In view of this, we exploit the big-M technique to decouple the reliability requirements. Based on the conclusion in \cite{SunITWC2019}, we adopt a near-optimal combination of decoding error probabilities for the two-hop transmission mode.
  \item In terms of the binary indicator constraints, we apply two efficient penalty approaches to resolve it. One is a novel NCP approach, which first equivalently transforms the binary variable constraints into continuous variable constraints plus $\ell_0$-norm constraints. The key of the NCP method is to represent $\ell_0$-norm constraints as the equivalent form of the difference of $\ell_1$-norm and $\ell_2$-norm and add this non-convex norm-based equality constraint as a penalty term into the objective function, thereby generating a penalized problem. The other one is a quadratic penalty method, which uses continuous indicator constraints and quadratic constraints to equivalently replace the original binary constraints and then adds this concave quadratic function as a penalty term into the objective function.
  \item Numerical results verify the efficiency of the applied NCP-based and QP-based SCA iterative algorithms for solving the formulated problem. In addition, the impacts of reliability, the number and location of relays, the number of robots, the target number of data bits on the total power consumption are revealed.
\end{itemize}
\section{System Model}
Consider an uplink relay-aided URLLC transmission in a smart factory with one controller, $N$ fixed DF relays operating in a half-duplex mode indexed by $n\in\mathcal{N}\triangleq\{1,\cdots,N\}$ and $K$ randomly distributed robots indexed by $k\in\mathcal{K}\triangleq\{1,\cdots,K\}$, as shown in Fig. \ref{systemModel}. All devices including the controller, relays, and robots are equipped with a single antenna. The available system spectrum consists of $M$ orthogonal RBs indexed by $m\in\mathcal{M}\triangleq\{1,\cdots,M\}$, each with bandwidth $W$Hz.
\begin{figure}[htbp]
	\centering
	\includegraphics[width=.9\linewidth]{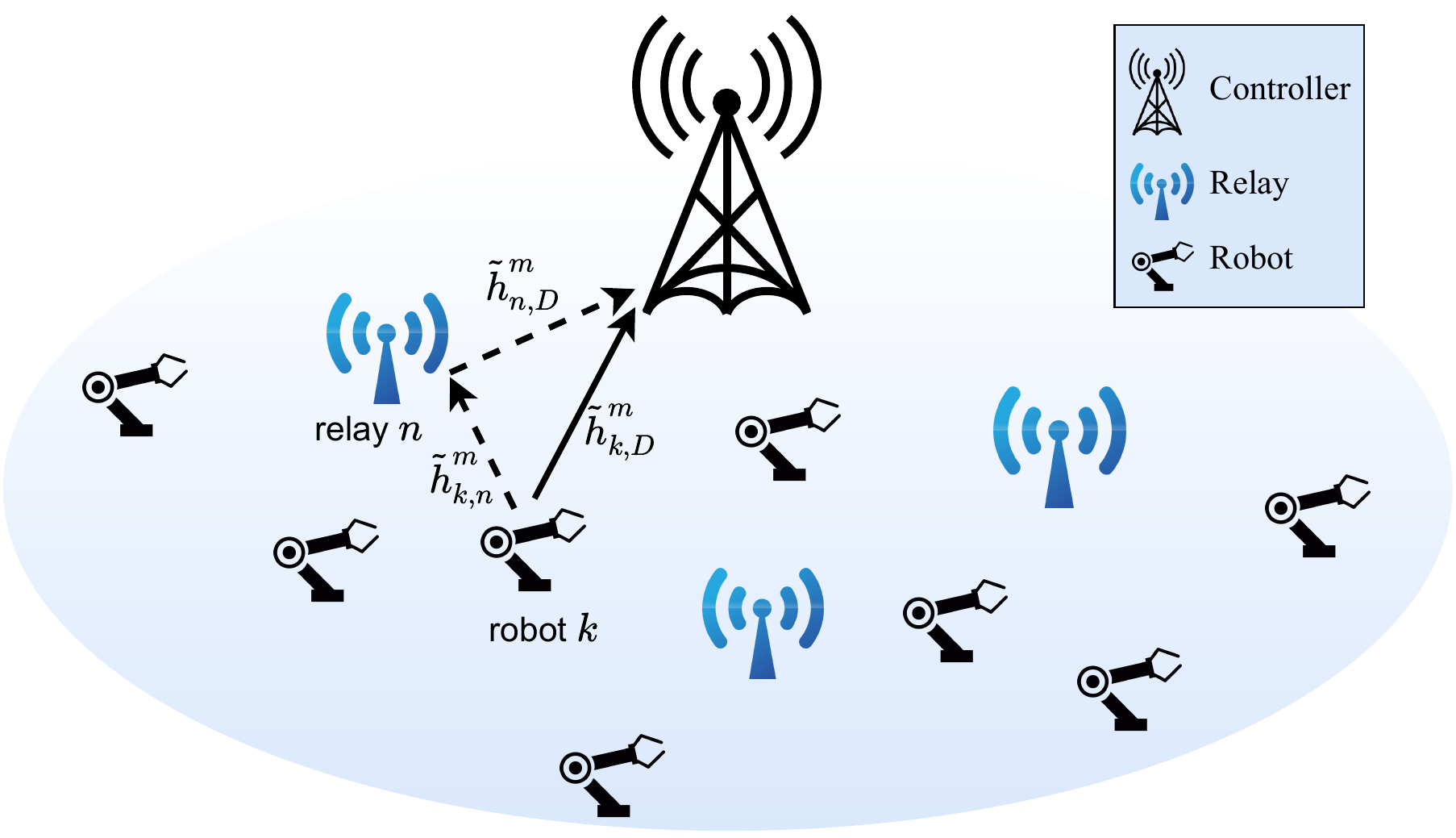}
	\caption{A relay-aided uplink URLLC system with one controller, multiple robots and relays.}
	\label{systemModel}
\end{figure}

Each robot $k$ has its own critical QoS requirements on the packet transmission, mainly including the minimum number of $B_k$ data bits and the allowable maximum packet error probability (PEP) $\varepsilon_{\max}$. If a robot experiences deep fading in a certain period of time, the critical QoS requirements cannot be guaranteed. Therefore, a relay-aided two-phase transmission protocol is proposed to secure each robot's ultra-reliable communication by means of spatial diversity. In this protocol, each robot has two types of transmission modes, i.e., each robot can deliver a packet to the controller directly under the well-conditioned channel or by the cooperative communication aided by one selected relay. In order to clearly indicate the relay selection and RB assignment, we denote the indicator by $\phi_{k,n}^m,n\in\mathcal{N}_0\triangleq \mathcal{N}\cup \{0\}$, where $\phi_{k,0}^m=1$ refers to the direct transmission from the robot to the controller, $\phi_{k,n}^m=1,\forall n\in\mathcal{N}$ means the $k$-th robot's packet is forwarded by the $n$-th relay on the $m$-th RB; otherwise, $\phi_{k,n}^m=0$. It is assumed that one RB is allocated to at most one robot whether its packet is directly delivered or forwarded through relay-assisted communication, such that
\begin{equation}
  \sum_{k=1}^K \sum_{n=0}^N \phi_{k,n}^m\leq 1,\forall m.
\end{equation}
Thus, there is no interference among each robot's packet transmission. Due to the short-packet transmission characteristic of URLLC systems and low data demands of mission-critical IoT applications, the allocation of multiple RBs to one robot is not efficient and such RBs cannot be fully utilized. In light of this, we assume each robot can only occupy one RB, namely,
\begin{equation}
  \sum_{m=1}^M \sum_{n=0}^N \phi_{k,n}^m=1,\forall k.
\end{equation}
\subsection{Achievable Rate Characterization of URLLC Packets}
Via the tight normal approximation of capacity under the AWGN channels, the maximum number of received data bits for each robot under the proposed relay-aided two-phase transmission protocol can be characterized.
\subsubsection{Phase I}
If the robot $k$ delivers a packet directly to the controller on the RB $m$, the maximum number of received data bits $R_{k,D}^m$ of the controller can be given by
\begin{equation}
  \!\!R_{k,D}^m\!=\!\tau_1 W\!\!\left(\log_2(1\!+\!h_{k,D}^mp_{k,D}^m)\!-\!\sqrt{\frac{V_{k,D}^m}{\tau_1 W}} \frac{Q^{-1}(\varepsilon_{k,D}^m)}{\ln 2}\!\!\right)\!,
\end{equation}
where $\tau_1$ is the duration of the first transmission phase, $h_{k,D}^m=\frac{|\tilde{h}_{k,D}^m|^2}{\sigma_{k,D}^2}$, $\tilde{h}_{k,D}^m\in\mathbb{C}$ is the channel from the $k$-th robot to the controller (destination node) on the $m$-th RB, $p_{k,D}^m$ is the transmit power, $V_{k,D}^m=1-\frac{1}{\left(1+h_{k,D}^mp_{k,D}^m\right)^2}$ is the channel dispersion, $\varepsilon_{k,D}^m$ is the decoding error probability, $Q^{-1}(\varepsilon_{k,D}^m)$ is the inverse of $Q(\varepsilon_{k,D}^m)=\int_{\varepsilon_{k,D}^m}^{\infty} \frac{1}{\sqrt{2 \pi}} e^{-t^{2}/2}dt$. Note that in this direct transmission mode, the robot $k$ is not allowed to use the subsequent idle second phase. In this case, the decoding error probability is the packet error probability, which has to satisfy
\begin{equation}
  \varepsilon_{k,D}^m\leq \varepsilon_{\max}.
\end{equation}

If the relay $n$ is used for the $k$-th robot's packet transmission on the $m$-th RB, the maximum number of received data bits $R_{k,n}^m$ of the $n$-th relay can be expressed as
\begin{equation}
  \!R_{k,n}^m\!=\!\tau_1 W\!\!\left(\log_2(1\!+\!h_{k,n}^mp_{k,n}^m)\!-\!\sqrt{\frac{V_{k,n}^m}{\tau_1 W}} \frac{Q^{-1}(\varepsilon_{k,n}^m)}{\ln 2}\!\right)\!,
\end{equation}
where $h_{k,n}^m=\frac{|\tilde{h}_{k,n}^m|^2}{\sigma_{k,n}^2}$, $\tilde{h}_{k,n}^m\in\mathbb{C}$ is the channel from the $k$-th robot to the $n$-th relay on the RB $m$, $p_{k,n}^m$ is the transmit power,   $V_{k,n}^m=1-\frac{1}{\left(1+h_{k,n}^mp_{k,n}^m\right)^2}$ is the channel dispersion, $\varepsilon_{k,n}^m$ is the decoding error probability.
\subsubsection{Phase II}
In the second phase, relay $n$ decodes the packet sent from robot $k$, regenerates a packet, and forwards it to the controller on the same RB using the regenerate-and-forward cooperative protocol. The maximum number of received data bits $R_{n,D}^{m,k}$ of the controller can be represented as
\begin{equation}
  \!\!R_{n,D}^{m,k}\!\!=\!\!\tau_2 W\!\!\left(\!\!\log_2(1\!+\!h_{n,D}^mp_{n,D}^{m,k})\!-\!\sqrt{\frac{V_{n,D}^{m,k}}{\tau_2 W}} \frac{Q^{-1}(\varepsilon_{n,D}^{m,k})}{\ln 2}\!\!\right)\!,
\end{equation}
where $\tau_2$ is the duration of the second transmission phase, $h_{n,D}^m=\frac{|\tilde{h}_{n,D}^m|^2}{\sigma_{n,D}^2}$, $\tilde{h}_{n,D}^m\in\mathbb{C}$ is the channel from the $n$-th relay to the controller on the $m$-th RB, $p_{n,D}^{m,k}$ is the transmit power,   $V_{n,D}^{m,k}=1-\frac{1}{\left(1+h_{n,D}^mp_{n,D}^{m,k}\right)^2}$ is the channel dispersion, $\varepsilon_{n,D}^{m,k}$ is the decoding error probability.

Note that if the $k$-th robot's packet transmission is aided by the $n$-th relay on the $m$-th RB, the overall packet error probability requirement is given by
\begin{align}
  1-(1-\varepsilon_{k,n}^m)(1-\varepsilon_{n,D}^{m,k})&=\varepsilon_{k,n}^m+\varepsilon_{n,D}^{m,k}
  -\varepsilon_{k,n}^m \varepsilon_{n,D}^{m,k}\notag\\
  &\overset{(a)}{\approx} \varepsilon_{k,n}^m+\varepsilon_{n,D}^{m,k}\leq \varepsilon_{\max},
\end{align}
where approximation $(a)$ is accurate enough since each decoding error probability is extremely small.
\subsection{Problem Formulation}
Considering that each robot has to complete the uplink transmission task of $B_k$ data bits, the overall throughput requirement for each robot $k$ under the relay-aided two-phase transmission protocol can be characterized by
\begin{equation}\label{capacity}
  \sum_{m=1}^M \sum_{n=1}^N \phi_{k,n}^m R_{n,D}^{m,k}+\sum_{m=1}^M \phi_{k,0}^m R_{k,D}^m\geq B_k,\forall k.
\end{equation}
It is necessary to remark that \eqref{capacity} holds with equality for the optimal solution. That is, robot $k$ either forwards $B_k$ data bits to the controller through one relay or directly transmits $B_k$ data bits to the controller.

Note that if relay $n$ is used for the $k$-th robot's packet transmission, the throughput of the link from the robot $k$ to the $n$-th relay has to be equal to or larger than that of the link from the relay $n$ to the controller. From \eqref{capacity}, the throughput of the link from the relay $n$ to the controller is $B_k$ data bits. Hence, we have the following constraint
\begin{equation}
  \sum_{m=1}^M \sum_{n=1}^N \phi_{k,n}^m R_{k,n}^m \geq \sum_{m=1}^M \sum_{n=1}^N \phi_{k,n}^m B_k,\forall k.
\end{equation}

Meanwhile, for each robot $k$, such target uplink transmission task has to be completed within the allowable maximum packet error probability $\varepsilon_{\max}$. Therefore, to ensure the overall reliability requirement, under the relay-aided two-phase transmission protocol, each decoding error probability component has to satisfy
\begin{equation}
  \sum_{m=1}^M \sum_{n=1}^N \phi_{k,n}^m(\varepsilon_{k,n}^m+\varepsilon_{n,D}^{m,k})+\sum_{m=1}^M \phi_{k,0}^m \varepsilon_{k,D}^m\leq \varepsilon_{\max},\forall k.
\end{equation}

In such a relay-aided two-phase transmission protocol, our aim is to design an efficient relay selection, RB assignment, and power control strategy under the finite blocklength regime to ensure all robots' ultra-reliable target payload transmission and minimize the total transmit power of all robots and relays. To be more explicit, the optimization problem can be formulated as
\begin{subequations}
\begin{align}
\!\!\!\mathrm{P1}:\min_{\Phi,\mathcal{P,E}} \quad& P_{\mathrm{tot}}\\
\mathrm{s.t.}\quad~ &\!\!\sum_{m=1}^M\! \sum_{n=1}^N\! \phi_{k,n}^m R_{n,D}^{m,k}\!+\!\!\!\sum_{m=1}^M\! \phi_{k,0}^m R_{k,D}^m\!\!\geq\!\!B_k,\!\forall k,\label{P1-C1}\\
&R_{k,n}^m \geq R_{n,D}^{m,k},\forall k,n,m,\label{P1-C2}\\
&\!\!\sum_{m=1}^M\! \sum_{n=1}^N\! \phi_{k,n}^m\!(\varepsilon_{k,n}^m\!+\!\varepsilon_{n,D}^{m,k})\!+\!\!\!\sum_{m=1}^M\! \phi_{k,0}^m \varepsilon_{k,D}^m\!\!\leq\!\!\varepsilon_{\max},\!\forall k,\label{P1-C3}\\
&\phi_{k,n}^m \in \{0,1\},\forall k,m,n\in\mathcal{N}_0,\label{P1-C5}\\
&\sum_{k=1}^K \sum_{n=0}^N \phi_{k,n}^m\leq 1,\forall m,\label{P1-C6}\\
&\sum_{m=1}^M \sum_{n=0}^N \phi_{k,n}^m=1,\forall k,\label{P1-C7}\\
&p_{k,D}^m,p_{k,n}^m,p_{n,D}^{m,k}\geq 0,\forall k,m,n,\label{P1-C8}
\end{align}
\end{subequations}
where
\begin{equation}
  P_{\mathrm{tot}}\!\!\triangleq\!\!\sum_{k=1}^K\!\sum_{m=1}^M\! \sum_{n=1}^N\! \phi_{k,n}^m\! \left(\!p_{k,n}^m\!+\!p_{n,D}^{m,k}\right)\!+\!\sum_{k=1}^K\!\sum_{m=1}^M\! \phi_{k,0}^m p_{k,D}^m,
\end{equation}
$\Phi=\{\phi_{k,n}^m,\forall k\in\mathcal{K},m\in\mathcal{M},n\in\mathcal{N}_0\}$ is the set of the relay selection and RB assignment indicators, $\mathcal{P}=\{p_{k,n}^m,p_{n,D}^{m,k},p_{k,D}^m,\forall k\in\mathcal{K},m\in\mathcal{M},n\in\mathcal{N}\}$ is the set of transmit powers, $\mathcal{E}=\{\varepsilon_{k,n}^m,\varepsilon_{n,D}^{m,k},\varepsilon_{k,D}^m,\forall k\in\mathcal{K},m\in\mathcal{M},n\in\mathcal{N}\}$ is the set of decoding error probabilities.

Constraint \eqref{P1-C1} means that for any robot $k$, it has to transmit $B_k$ data bits whether the $k$-th robot's packet is directly delivered to the controller or is forwarded to the controller by one relay. Constraint \eqref{P1-C2} requires that if one relay is used, the capacity of the robot-relay link should be larger enough to transmit $B_k$ data bits within the transmission duration. The overall packet error probability requirement is reflected in the constraint \eqref{P1-C3}. Constraints \eqref{P1-C5}-\eqref{P1-C7} refer to the relay selection and RB assignment strategy and indicate that one RB is allocated to at most one robot and one relay, and one robot can only occupy one RB and one relay (if necessary). The constraint \eqref{P1-C8} is the nonnegative transmit power constraint.

The problem $\mathrm{P1}$ is a mixed-integer strictly non-convex problem. One challenging aspect is the well-known discrete nature of the binary relay selection and RB assignment indicator. Secondly, three sets of optimization variables are coupled together, which makes the problem even harder to solve. Thirdly, the expressions $R_{k,n}^m,R_{n,D}^{m,k},R_{k,D}^m$ are highly intractable and are neither convex nor concave w.r.t.  $p_{k,n}^m,p_{n,D}^{m,k}$ and $p_{k,D}^m$, respectively. In general, this problem is NP-hard and its globally optimal solution cannot be efficiently obtained by the exhaustive search method with high computation cost. Therefore, a low-complexity algorithm is needed to obtain a sub-optimal solution to the problem $\mathrm{P1}$, thereby efficiently scheduling the URLLC packets' ultra-reliable transmission.
\section{Energy-efficient Relay Selection and Resource Allocation Optimization}
\subsection{Problem Reformulation}
Using the approximations that $V_{k,n}^m\approx 1,V_{n,D}^{m,k}\approx 1,V_{k,D}^m\approx 1$, which are tight at moderate or large received SNR, i.e., higher than $3$ dB, such SNR condition is almost true for URLLC packet scheduling. Meanwhile, approximating the channel dispersion to one has been well-acknowledged in most of the current research works \cite{SheITWC2018,RenITC2020}. In essence, a lower bound of $R_{k,n}^m,R_{n,D}^{m,k},R_{k,D}^m$ leads to more stringent requirements. Thus, $R_{k,n}^m,R_{n,D}^{m,k},R_{k,D}^m$ can be approximated by
\begin{align}
  &\bar{R}_{k,n}^m=\tau_1W\left(\log_2(1+h_{k,n}^mp_{k,n}^m)
  -\frac{Q^{-1}(\varepsilon_{k,n}^m)}{\sqrt{\tau_1W}\ln 2}\right),\\
  &\bar{R}_{n,D}^{m,k}=\tau_2W\left(\log_2(1+h_{n,D}^mp_{n,D}^{m,k})
  -\frac{Q^{-1}(\varepsilon_{n,D}^{m,k})}{\sqrt{\tau_2W}\ln 2}\right),
\end{align}
and
\begin{equation}
  \bar{R}_{k,D}^m=\tau_1W\left(\log_2(1+h_{k,D}^mp_{k,D}^m)-\frac{Q^{-1}(\varepsilon_{k,D}^m)}{\sqrt{\tau_1W}\ln 2}\right),\\
\end{equation}
respectively.

Therefore, the problem $\mathrm{P1}$ is reformulated as
\begin{subequations}
\begin{align}
\mathrm{P2}:\min_{\Phi,\mathcal{P,E}} \quad& P_{\mathrm{tot}}\\
\mathrm{s.t.}\quad~&\!\!\sum_{m=1}^M \!\sum_{n=1}^N\! \phi_{k,n}^m \bar{R}_{n,D}^{m,k}\!+\!\!\!\sum_{m=1}^M \! \phi_{k,0}^m \bar{R}_{k,D}^m\!\!\geq\!\! B_k,\!\forall k,\label{P2-C1}\\
&\sum_{m=1}^M\! \sum_{n=1}^N\! \phi_{k,n}^m \bar{R}_{k,n}^m\! \geq \!\!\sum_{m=1}^M \!\sum_{n=1}^N \! \phi_{k,n}^m B_k,\forall k,\label{P2-C2}\\
&\eqref{P1-C3}\!-\!\eqref{P1-C8}.
\end{align}
\end{subequations}

The approximated problem $\mathrm{P2}$ is still challenging to solve due to the coupled nature of optimization variables, i.e., $\Phi,\mathcal{P}$ and $\Phi,\mathcal{E}$ are coupled in constraints \eqref{P2-C1},\eqref{P2-C2},\eqref{P1-C3} and the binary nature of the relay selection and RB assignment indicator reflected in constraint \eqref{P1-C5}.

First, let us deal with the problem of $\phi_{k,n}^m$ coupled with $p_{k,n}^m,p_{n,D}^{m,k},p_{k,D}^m$, which exists in the objective function and constraints \eqref{P2-C1}-\eqref{P2-C2}. Towards this end, we introduce new slack variables as follows
\begin{equation}\label{slack}
  \tilde{p}_{k,n}^m=\phi_{k,n}^m p_{k,n}^m,~\tilde{p}_{n,D}^{m,k}=\phi_{k,n}^mp_{n,D}^{m,k},~\tilde{p}_{k,D}^m=\phi_{k,0}^mp_{k,D}^m,
\end{equation}
for all $k,m,n$, and the objective function can be accordingly written as
\begin{equation}
  \widetilde{P}_{\mathrm{tot}}=\sum_{k=1}^K\sum_{m=1}^M \sum_{n=1}^N \left(\tilde{p}_{k,n}^m+\tilde{p}_{n,D}^{m,k}\right)+\sum_{k=1}^K\sum_{m=1}^M \tilde{p}_{k,D}^m.
\end{equation}
By leveraging the properties of the relative entropy function, the terms $\phi_{k,n}^m R_{n,D}^{m,k},~\phi_{k,0}^m R_{k,D}^m$ in the constraint \eqref{P1-C1} can be correspondingly represented as
\begin{equation}\label{R_nDmk}
  \widetilde{R}_{n,D}^{m,k}\!=\!
  \tau_2W\!\!\left(\!\!\phi_{k,n}^m\log_2\!\left(\!1\!+
  \!\frac{h_{n,D}^m\tilde{p}_{n,D}^{m,k}}{\phi_{k,n}^m}\!\right)
  \!-\!\frac{\phi_{k,n}^mQ^{-1}(\varepsilon_{n,D}^{m,k})}{\sqrt{\tau_2W}\ln 2}\!\right)\!,
\end{equation}
\begin{equation}\label{R_kDm}
  \widetilde{R}_{k,D}^m\!=\!\tau_1W\!\!\left(\!\!\phi_{k,0}^m\log_2\!\left(\!1\!+
  \!\frac{h_{k,D}^m\tilde{p}_{k,D}^m}{\phi_{k,0}^m}\!\right)
  \!-\!\frac{\phi_{k,0}^mQ^{-1}(\varepsilon_{k,D}^m)}{\sqrt{\tau_1W}\ln 2}\!\right).
\end{equation}
Also, the term $\phi_{k,n}^m R_{k,n}^m$ in the constraint \eqref{P1-C2} can be alternatively expressed as
\begin{equation}\label{R_knm}
  \widetilde{R}_{k,n}^m\!=\!\tau_1W\!\!\left(\!\phi_{k,n}^m\log_2\!\left(\!1\!+
  \!\frac{h_{k,n}^m\tilde{p}_{k,n}^m}{\phi_{k,n}^m}\!\right)
  \!-\!\frac{\phi_{k,n}^mQ^{-1}(\varepsilon_{k,n}^m)}{\sqrt{\tau_1W}\ln 2}\!\right).
\end{equation}
\begin{remark}\label{remark1}
Note that the introduction of slack variables $\tilde{p}_{k,n}^m,\tilde{p}_{n,D}^{m,k},\tilde{p}_{k,D}^m$ in \eqref{slack} means that if $\phi_{k,n}^m=0$, then $\tilde{p}_{k,n}^m=0,\tilde{p}_{n,D}^{m,k}=0$, and if $\phi_{k,0}^m=0$, then $\tilde{p}_{k,D}^m=0$. Taking $\phi_{k,n}^m R_{k,n}^m$ as an example for explanation. The common method is to reformulate $\phi_{k,n}^m R_{k,n}^m$ as
\begin{equation}
  \!\phi_{k,n}^m R_{k,n}^m\!=\!\tau_1W\!\!\left(\!\log_2\!(\!1\!+\!h_{k,n}^m\tilde{p}_{k,n}^m)
  \!-\!\frac{\phi_{k,n}^mQ^{-1}(\varepsilon_{k,n}^m)}{\sqrt{\tau_1W}\ln 2}\!\right),
\end{equation}
using a constraint to limit the relationship between the variable $\phi_{k,n}^m$ and $\tilde{p}_{k,n}^m$, i.e., satisfying $\tilde{p}_{k,n}^m=0$ when $\phi_{k,n}^m=0$. Here, we exploit the properties of the relative entropy function $\mathrm{rel\_entr}(x,y)=x\log\left(\frac{x}{y}\right)$ that $\mathrm{rel\_entr}(x,y)$ is a convex function w.r.t. $x,y$ and it equals to $0$ when $x=0,y\geq0$, thus generating the reformulations \eqref{R_nDmk}-\eqref{R_knm}. Meanwhile, the objective function is the total power minimization, we have $\tilde{p}_{k,n}^m=0,\tilde{p}_{n,D}^{m,k}=0$ when $\phi_{k,n}^m=0$ and $\tilde{p}_{k,D}^m=0$ when $\phi_{k,0}^m=0$. By virtue of such two factors, the transformation reflected in \eqref{R_nDmk}-\eqref{R_knm} are equivalent to the conventional method, and the advantage of the proposed transformations is to avoid adding extra constraints.
\end{remark}

Based on the above transformations, the problem $\mathrm{P2}$ can be reformulated as
\begin{subequations}
\begin{align}
\mathrm{P3}:\min_{\Phi,\mathcal{\widetilde{P},E}} \quad& \widetilde{P}_{\mathrm{tot}}\\
\mathrm{s.t.}\quad~ &\sum_{m=1}^M \sum_{n=1}^N \widetilde{R}_{n,D}^{m,k}+\sum_{m=1}^M \widetilde{R}_{k,D}^m\geq B_k,\forall k,\label{P3-C1}\\
&\sum_{m=1}^M \sum_{n=1}^N \widetilde{R}_{k,n}^m \geq \sum_{m=1}^M \sum_{n=1}^N \phi_{k,n}^m B_k,\forall k,\label{P3-C2}\\
&\tilde{p}_{k,D}^m,\tilde{p}_{k,n}^m,\tilde{p}_{n,D}^{m,k}\geq 0,\forall k,m,n,\label{P3-C5}\\
&\eqref{P1-C3}\!-\!\eqref{P1-C7},
\end{align}
\end{subequations}
where $\mathcal{\widetilde{P}}=\{\tilde{p}_{k,n}^m,\tilde{p}_{n,D}^{m,k},\tilde{p}_{k,D}^m,\forall k\in\mathcal{K},m\in\mathcal{M},n\in\mathcal{N}\}$ is the set of transmit powers.

\begin{lemma}
\label{lemma1}
  The optimal solutions to problem $\mathrm{P3}$ are optimal to problem $\mathrm{P2}$.
\end{lemma}
\begin{proof}
The feasible set of the problem $\mathrm{P2}$ contains that of the problem $\mathrm{P3}$. Specifically, when $\phi_{k,n}^m=0,n\in\mathcal{N}_0$, the corresponding $p_{k,n}^m,p_{n,D}^{m,k},p_{k,D}^m$ in the problem $\mathrm{P2}$ can take any value which can be intuitively seen from the objective function, while $\tilde{p}_{k,n}^m,\tilde{p}_{n,D}^{m,k},\tilde{p}_{k,D}^m$ in the problem $\mathrm{P3}$ can only take zero. When $\phi_{k,n}^m=1$, the problem $\mathrm{P2}$ is equivalent to the problem $\mathrm{P3}$. In fact, for the optimal solutions to the problem $\mathrm{P2}$, we can construct the optimal solutions to the problem $\mathrm{P3}$ by setting $p_{k,n}^m=0,p_{n,D}^{m,k}=0,p_{k,D}^m=0$ if $\phi_{k,n}^m=0,n\in\mathcal{N}_0$. Hence, the optimal solutions to the problem $\mathrm{P3}$ is optimal to the problem $\mathrm{P2}$.
\end{proof}

Based on Lemma \ref{lemma1}, the next focus is to solve the problem $\mathrm{P3}$. First, let us assress the coupling problem of optimization variables $\phi_{k,n}^m,n\in\mathcal{N}_0$ and $\varepsilon_{k,n}^m,\varepsilon_{n,D}^{m,k},\varepsilon_{k,D}^m$ in the constraint \eqref{P1-C3}. For the constraint \eqref{P1-C3}, we can exploit the big-M technique to deal with it. More explicitly, the constraint \eqref{P1-C3} can be equivalently decoupled into the following two convex constraints as
\begin{align}
  &\varepsilon_{k,n}^m+\varepsilon_{n,D}^{m,k}\leq \varepsilon_{\max}+1-\phi_{k,n}^m,\forall k,m,n,\label{PER-relay}\\
  &\varepsilon_{k,D}^m\leq \varepsilon_{\max}+1-\phi_{k,0}^m,\forall k,m.\label{PER-directTransmission}
\end{align}

Note that the constraint \eqref{PER-relay} always holds when no relay is activated ($\phi_{k,n}^m=0$), thereby only posing a reliability limit to the two-hop link transmission when one relay is activated, i.e., $\varepsilon_{k,n}^m+\varepsilon_{n,D}^{m,k}\leq \varepsilon_{\max}$ for $\mu_{k,n}^m=1$. Furthermore, the maximum number of received data bits $\widetilde{R}_{k,n}^m,\widetilde{R}_{n,D}^{m,k}$ are strictly increasing w.r.t. $\varepsilon_{k,n}^m,\varepsilon_{n,D}^{m,k}$, respectively, thereby reducing power consumption. So the inequality holds with the equality for the optimal solutions $\varepsilon_{k,n}^m,\varepsilon_{n,D}^{m,k}$, i.e., $\varepsilon_{k,n}^m+\varepsilon_{n,D}^{m,k}=\varepsilon_{\max}$. As discussed in \cite{SunITWC2019}, compared with finding the optimal combination of such two decoding error probabilities, setting these probability terms equal only causes minor performance loss in transmit power. More importantly, setting them equal will considerably simplify the optimization of the relay selection, RB assignment and power control. Hence, we adopt a near optimal combination of decoding error probabilities that guarantees the overall packet error probability, i.e.,
\begin{equation}\label{two-topReliability}
  \varepsilon_{k,n}^m=\varepsilon_{n,D}^{m,k}=\varepsilon_{\max}/2,\forall k,m,n.
\end{equation}

Likewise, the constraint \eqref{PER-directTransmission} always holds for $\mu_{k,n}^m=0$ and when the direct transmission happens, there exists an upper bound for the allowable decoding error probability, i.e., $\varepsilon_{k,D}^m\leq\varepsilon_{\max}$ for $\mu_{k,n}^m=1$. In view that $\widetilde{R}_{k,D}^m$ is strictly increasing in $\varepsilon_{k,D}^m$, larger $\varepsilon_{k,D}^m$ means larger achievable rate and lower power consumption. Thus, the equality holds true for the optimal solution $\varepsilon_{k,D}^m$, i.e.,
\begin{equation}\label{directReliability}
  \varepsilon_{k,D}^m=\varepsilon_{\max},\forall k,m.
\end{equation}

The above near optimal solutions $\varepsilon_{k,n}^m,\varepsilon_{n,D}^{m,k},\forall k,m,n$ and optimal solutions $\varepsilon_{k,D}^m,\forall k,m$ significantly simplify the analysis and processing of the constraints \eqref{P3-C1} and \eqref{P3-C2}. Owing to them, we can readily analyze the convexity of constraints \eqref{P3-C1} and \eqref{P3-C2}. Since expressions $\widetilde{R}_{k,n}^m,\widetilde{R}_{n,D}^{m,k},\widetilde{R}_{k,D}^m$ in constraints \eqref{P3-C1} and \eqref{P3-C2} have the similar form, here, we take $\widetilde{R}_{k,n}^m$ as an example to interpret the concavity of $\widetilde{R}_{k,n}^m$ in detail. Note that for the expression $\widetilde{R}_{k,n}^m$, the first term $\phi_{k,n}^m\log_2\left(1+\frac{h_{k,n}^m\tilde{p}_{k,n}^m}{\phi_{k,n}^m}\right)$ inside it is a concave function where the concavity can be obtained based on the properties of the
relative entropy function. To be explicit, the first term can be equivalently represented as $-\phi_{k,n}^m\log_2\left(\frac{\phi_{k,n}^m}{\phi_{k,n}^m+h_{k,n}^m\tilde{p}_{k,n}^m}\right)$, which is a negative relative entropy function. Hence, the first term is concave function. Meanwhile, the second term inside it is a linear function since the near-optimal solution of the optimization variable $\varepsilon_{k,n}^m$ has been solved. Accordingly, it is easily obtained that $\widetilde{R}_{k,n}^m$ is a concave function w.r.t.  $\mu_{k,n}^m,\tilde{p}_{k,n}^m$. In a similar fashion, it can be conclude that $\widetilde{R}_{n,D}^{m,k}$ and $\widetilde{R}_{k,D}^m$ are concave functions. Thus, constraints \eqref{P3-C1} and \eqref{P3-C2} are convex constraints.

The remaining challenge in the problem $\mathrm{P3}$ is the binary relay selection and RB assignment variables $\phi_{k,n}^m \in \{0,1\},\forall k,m,n\in\mathcal{N}_0$ in the constraint \eqref{P1-C5} since the objective function and other constraints are transformed into convex ones by the above transformations. For the binary indicators, we first relax them to continuous variables, i.e., $\phi_{k,n}^m\in[0,1],\forall k,m,n\in\mathcal{N}_0$. Note that this relaxation is too loose and the feasible solutions may not be guaranteed. In order to obtain the sub-optimal sparse scheduling solutions, we apply the following two penalty methods.
\subsubsection{Non-convex Penalty Approach}
Considering that one RB is allocated to at most one robot and one relay reflected in the constraint \eqref{P1-C6}, we add an additional constraint
\begin{equation}
  \|\mathrm{vec}(\mathbf{S}_m)\|_0\leq 1,\forall m, \label{Sm}
\end{equation}
to guarantee the sparsity requirement on each RB, where $\mathbf{S}_m\in\mathbb{R}^{K\times (N+1)}_{+}$ is defined as
\begin{equation}\label{Sm}
  \mathbf{S}_m=
  \begin{bmatrix}
  \phi_{1,0}^m &\phi_{1,1}^m &\cdots &\phi_{1,N}^m\\
  \vdots & \vdots & \vdots & \vdots\\
  \phi_{K,0}^m &\phi_{K,1}^m &\cdots &\phi_{K,N}^m\\
  \end{bmatrix}.
\end{equation}

Meanwhile, given that one robot can only occupy one RB and one relay (if necessary) required in the constraint \eqref{P1-C7}, we have to add the following additional constraint
\begin{equation}
  \|\mathrm{vec}(\mathbf{U}_k)\|_0= 1,\forall k,\label{Uk}
\end{equation}
where $\mathbf{U}_k\in\mathbb{R}^{M\times (N+1)}_{+}$ is given by
\begin{equation}\label{Uk}
  \mathbf{U}_k=
  \begin{bmatrix}
  \phi_{k,0}^1 &\phi_{k,1}^1 &\cdots &\phi_{k,N}^1\\
  \vdots & \vdots & \vdots & \vdots\\
  \phi_{k,0}^M &\phi_{k,1}^M &\cdots &\phi_{k,N}^M\\
  \end{bmatrix}.
\end{equation}

Therefore, the problem $\mathrm{P3}$ can be reformulated in the following equivalent form as
\begin{subequations}
\begin{align}
\!\!\mathrm{P4\!\!-\!\!NCP}:\min_{\Phi,\mathcal{\tilde{P},E}} \quad& \widetilde{P}_{\mathrm{tot}}\\
\mathrm{s.t.}\quad~ &\!\eqref{P1-C6},\!\eqref{P1-C7},\!\eqref{P3-C1}\!\!-\!\!\eqref{P3-C5},
\eqref{two-topReliability},\!\eqref{directReliability},\label{P4-C1}\\
&\phi_{k,n}^m\in[0,1],\forall k,m,n\in\mathcal{N}_0,\label{P4-C2}\\
&\|\mathrm{vec}(\mathbf{S}_m)\|_0\leq 1,\forall m,\label{P4-C3}\\
&\|\mathrm{vec}(\mathbf{U}_k)\|_0= 1,\forall k.\label{P4-C4}
\end{align}
\end{subequations}

The novel NCP approach is tailored to deal with the $\ell_0$-norm constraints, hence constraints \eqref{P4-C3}-\eqref{P4-C4} in the problem $\mathrm{P4\!\!-\!\!NCP}$ can be efficiently resolved. More specifically, the principle of the NCP method is given in detail as follows. For any vector $\mathbf{x}\in\mathbb{R}^{n}$, it has at most one non-zero element if and only if
\begin{equation}
  \|\mathbf{x}\|_a=\|\mathbf{x}\|_b,~1\leq a<b,
\end{equation}
where $a,b$ are real and $\|\cdot\|_a$ and $\|\cdot\|_b$ represent $\ell_a$-norm and $\ell_b$-norm, respectively. Alternatively, by adding a power exponent $v>0$, we have
\begin{equation}
    \|\mathbf{x}\|_a^v=\|\mathbf{x}\|_b^v,~1\leq a<b.
\end{equation}
Hence, the constraint \eqref{P4-C3} can be expressed as the following equivalent form
\begin{equation}
  \|\mathrm{vec}(\mathbf{S}_m)\|_a^v=\|\mathrm{vec}(\mathbf{S}_m)\|_b^v,~1\leq a<b,\forall m.
\end{equation}

Actually, the constraint \eqref{P4-C4} can be relaxed into an inequality constraint, i.e.,
\begin{equation}\label{inequalUk}
  \|\mathrm{vec}(\mathbf{U}_k)\|_0\leq 1,\forall k,
\end{equation}
and we adopt an operation similar to addressing the constraint \eqref{P4-C3}. The reason is that combing with the target payload constraint \eqref{P3-C1}, the inequality constraint \eqref{inequalUk} will always hold with equality. Thus, the constraint \eqref{P4-C4} can be written as
\begin{equation}
  \|\mathrm{vec}(\mathbf{U}_k)\|_a^v=\|\mathrm{vec}(\mathbf{U}_k)\|_b^v,~1\leq a<b,\forall k.
\end{equation}

Note that, in general, for some $1\leq a<b$, $\|\mathrm{vec}(\mathbf{S}_m)\|_a^v-\|\mathrm{vec}(\mathbf{S}_m)\|_b^v\geq 0$  and $\|\mathrm{vec}(\mathbf{U}_k)\|_a^v-\|\mathrm{vec}(\mathbf{U}_k)\|_b^v\geq 0$ always hold. In this paper, we adopt a smooth penalty by setting $a=1,b=2,v=2$.

To promote the sparsity that $\mathrm{vec}(\mathbf{S}_m)$ and $\mathrm{vec}(\mathbf{U}_k)$ have at most one non-zero element, we add a penalty term to the objective function, which is given by
\vspace{-2mm}
\begin{align}
  \!\!\mathcal{F}_1(\mathbf{S}_m,\mathbf{U}_k)\!&=\!\frac{\lambda}{2}\!\left(\sum_{m=1}^M\! \left(\|\mathrm{vec}(\mathbf{S}_m)\|_1^2\!-\!\|\mathrm{vec}(\mathbf{S}_m)\|_2^2\right)\right.\notag\\
  &~~~~~~~\left.+\!\sum_{k=1}^K\! \left(\|\mathrm{vec}(\mathbf{U}_k)\|_1^2\!-\!\|\mathrm{vec}(\mathbf{U}_k)\|_2^2\right)\!\right),
\end{align}
where $\lambda>0$ is a penalty factor. Note that the penalty $\mathcal{F}_1(\mathbf{S}_m,\mathbf{U}_k)$ is non-convex since the two terms included in it are in the form of difference-of-convex function. In view of this, the SCA method is adopted, which can iteratively approximate a non-convex function into a convex one \cite{Boyd2004}. Specifically, the first-order Taylor approximations are applied to the convex function $\|\mathrm{vec}(\mathbf{S}_m)\|_2^2$ and $\|\mathrm{vec}(\mathbf{U}_k)\|_2^2$, which are given by
\vspace{-2mm}
\begin{align}
  &\|\mathrm{vec}(\mathbf{S}_m)\|_2^2\!\approx\! 2 \left(\mathrm{vec}(\mathbf{S}_m^{(i)})\right)^T\!\!\mathrm{vec}(\mathbf{S}_m)\!-
  \!\|\mathrm{vec}(\mathbf{S}_m^{(i)})\|_2^2,\\
  &\!\|\mathrm{vec}(\mathbf{U}_k)\|_2^2\!\approx\! 2 \left(\mathrm{vec}(\mathbf{U}_k^{(i)})\right)^T\!\!\mathrm{vec}(\mathbf{U}_k)\!-
  \!\|\mathrm{vec}(\mathbf{U}_k^{(i)})\|_2^2,
\end{align}
where $\mathbf{S}_m^{(i)},\mathbf{U}_k^{(i)}$ are the values of $\mathbf{S}_m,\mathbf{U}_k$ in the $i$th iteration, respectively. Accordingly, the penalty term can be approximated by
\begin{align}
  &\!\!\mathcal{\widetilde{F}}_1(\mathbf{S}_m,\mathbf{U}_k)\approx\notag\\
  &\!\!\frac{\lambda}{2}\!\left(\sum_{m=1}^M\!\! \left(\!\|\mathrm{vec}(\mathbf{S}_m)\|_1^2\!-\!
2\!\left(\!\mathrm{vec}(\mathbf{S}_m^{(i)})\right)^T\!\!\mathrm{vec}(\mathbf{S}_m)
\!+\!\|\mathrm{vec}(\mathbf{S}_m^{(i)})\|_2^2\!\right)\!\right.\notag\\
  &\!\!+\!\left.\!\sum_{k=1}^K\!\! \left(\!\|\mathrm{vec}(\mathbf{U}_k)\|_1^2\!-\!
2\!\left(\!\mathrm{vec}(\mathbf{U}_k^{(i)})\right)^T\!\!\mathrm{vec}(\mathbf{U}_k)
\!+\!\|\mathrm{vec}(\mathbf{U}_k^{(i)})\|_2^2\!\right)\!\!\right)\!.
\end{align}

Based on this NCP approach, we can obtain the following penalized convex problem
\begin{subequations}
\begin{align}
\!\!\!\mathrm{P5\!\!-\!\!NCP}:\min_{\Phi,\mathcal{\widetilde{P},E}} \quad& \widetilde{P}_{\mathrm{tot}}+\mathcal{\widetilde{F}}_1(\mathbf{S}_m,\mathbf{U}_k)\\
\mathrm{s.t.}\quad~ &\!\eqref{P1-C6},\!\eqref{P1-C7},\!\eqref{P3-C1}\!\!-\!\!\eqref{P3-C5},
\eqref{two-topReliability},\!\eqref{directReliability},\!\eqref{P4-C2}.\label{P5-C1}
\end{align}
\end{subequations}
\subsubsection{Quadratic Penalty Approach}
Note that the binary relay selection and RB assignment constraint \eqref{P1-C5} can be expressed in the following equivalent form
\begin{align}
  &\phi_{k,n}^m\in[0,1],\forall k,m,n\in\mathcal{N}_0,\\
  &\sum_{k=1}^K \sum_{m=1}^M \sum_{n=0}^N \left(\phi_{k,n}^m-(\phi_{k,n}^m)^2\right)\leq 0.
\end{align}

Correspondingly, the problem $\mathrm{P3}$ can be equivalently reformulated as
\begin{subequations}
\begin{align}
\!\!\mathrm{P4\!\!-\!\!QP}:\min_{\Phi,\mathcal{\widetilde{P},E}} \quad& \widetilde{P}_{\mathrm{tot}}\\
\mathrm{s.t.}\quad~ &\eqref{P1-C6},\!\eqref{P1-C7},\!\eqref{P3-C1}\!\!-\!\!\eqref{P3-C5},\!
\eqref{two-topReliability},\!\eqref{directReliability},\!\label{P4-QP-C1}\\
&\phi_{k,n}^m\in[0,1],\forall k,m,n\in\mathcal{N}_0,\label{P4-QP-C2}\\
&\sum_{k=1}^K \sum_{m=1}^M \sum_{n=0}^N \left(\phi_{k,n}^m-(\phi_{k,n}^m)^2\right)\leq 0,\label{P4-QP-C3}
\end{align}
\end{subequations}
In the problem $\mathrm{P4\!\!-\!\!QP}$, the only non-convexity is derived from the constraint \eqref{P4-QP-C3}, where the left-hand side is a concave function. To deal with this constraint, we give the following Theorem \ref{theorem}.
\begin{theorem}\label{theorem}
  For a sufficiently large constant $\beta\gg 1$, the optimization problem $\mathrm{P4\!\!-\!\!QP}$ is equivalent to the following problem
\begin{subequations}
\begin{align}
\!\!\mathrm{P5\!\!-\!\!QP}:\min_{\Phi,\mathcal{\widetilde{P},E}} \quad& \widetilde{P}_{\mathrm{tot}}+\mathcal{F}_2(\phi_{k,n}^m)\label{P5-QP-obj}\\
\mathrm{s.t.}\quad~ &\eqref{P1-C6},\!\eqref{P1-C7},\!\eqref{P3-C1}\!\!-\!\!\eqref{P3-C5},\!
\eqref{two-topReliability},\!\eqref{directReliability},\!\label{P5-QP-C1}\\
&\phi_{k,n}^m\in[0,1],\forall k,m,n\in\mathcal{N}_0,\label{P5-QP-C2}
\end{align}
\end{subequations}
where
\begin{equation}
  \mathcal{F}_2(\phi_{k,n}^m)=\beta\sum_{k=1}^K \sum_{m=1}^M \sum_{n=0}^N \left(\phi_{k,n}^m-(\phi_{k,n}^m)^2\right),
\end{equation}
is a quadratic penalty term, and $\beta$, as a penalty factor, penalizes the objective function for any $\phi_{k,n}^m$ that is not equal to $0$ or $1$.
\end{theorem}
\begin{proof}
  The proof of the Theorem \ref{theorem} can be referred to the Appendix \ref{appendixA}.
\end{proof}

The transformed equivalent problem $\mathrm{P5\!\!-\!\!QP}$ is still non-convex since the quadratic penalty term in \eqref{P5-QP-obj} is a concave function. By virtue of the first-order Taylor approximation, we can approximate $(\phi_{k,n}^m)^2$ by
\begin{equation}
  (\phi_{k,n}^m)^2\approx 2(\phi_{k,n}^m)^{(i)}\phi_{k,n}^m-\left((\phi_{k,n}^m)^{(i)}\right)^2,
\end{equation}
where $(\phi_{k,n}^m)^{(i)}$ is the value of the variable $\phi_{k,n}^m$ in the $i$th iteration. Therefore, the penalty term can be approximated by
\begin{equation}
  \!\!\mathcal{\widetilde{F}}_2(\phi_{k,n}^m)\!\approx\!\!\beta\!\sum_{k=1}^K\! \sum_{m=1}^M \!\sum_{n=0}^N\! \left(\!\!\phi_{k,n}^m\!-\!2(\phi_{k,n}^m)^{(i)}\phi_{k,n}^m\!+\!\left(\!(\phi_{k,n}^m)^{(i)}\!\right)^2\!\right).
\end{equation}

The problem $\mathrm{P5\!\!-\!\!QP}$ can be reformulated as the following penalized convex problem
\begin{subequations}
\begin{align}
\!\!\mathrm{P6\!\!-\!\!QP}:\min_{\Phi,\mathcal{\widetilde{P},E}} \quad& \widetilde{P}_{\mathrm{tot}}+\mathcal{\widetilde{F}}_2(\phi_{k,n}^m)\label{P6-QP-obj}\\
\mathrm{s.t.}\quad~ &\eqref{P1-C6},\!\eqref{P1-C7},\!\eqref{P3-C1}\!\!-\!\!\eqref{P3-C5},\!
\eqref{two-topReliability},\!\eqref{directReliability},\!\eqref{P5-QP-C2}.\label{P6-QP-C1}
\end{align}
\end{subequations}
\subsection{Proposed Algorithm}
\subsubsection{Non-convex Penalty Approach}
According to the preceding analysis, we propose an NCP-based SCA iterative algorithm to solve the problem $\mathrm{P2}$. The algorithm is summarized in Algorithm \ref{Alg1}. Note that a stationary-point solution of the problem $\mathrm{P2}$ can be obtained by solving a sequence of the penalized convex problem $\mathrm{P5\!\!-\!\!NCP}$ until convergence \cite{Wang2019}.
\begin{algorithm}[htbp]
\caption{\textbf{:} NCP-based SCA Iterative Algorithm for Solving Problem $\mathrm{P2}$}
\label{Alg1}
\begin{spacing}{1.3}
	\begin{algorithmic}[1]
		\STATE \textbf{Initialize}~iteration index $i=0$, feasible $(\phi_{k,n}^m)^{(0)},\forall k,m,n\in\mathcal{N}_0$, initial penalty factor $\lambda^{(0)}>0$, $\eta>1$, tolerance $\epsilon>0$.
        \STATE Obtain $\mathbf{S}_m^{(0)},\mathbf{U}_k^{(0)}$ based on equations \eqref{Sm} and \eqref{Uk}, respectively.
		\REPEAT
        \STATE Set $i=i+1$.
		\STATE Obtain $\{\phi_{k,n}^m,\tilde{p}_{k,n}^m,\tilde{p}_{n,D}^{m,k},\tilde{p}_{k,D}^m\}$ by solving $\mathrm{P5\!\!-\!\!NCP}$ and restore $\tilde{p}_\mathrm{tot}^{(i)}$.
        \STATE Update $(\phi_{k,n}^m)^{(i)}=\phi_{k,n}^m,\forall k,m,n\in\mathcal{N}_0$.
        \STATE Update $\mathbf{S}_m^{(i)},\mathbf{U}_k^{(i)}$ based on equations \eqref{Sm} and \eqref{Uk}, respectively.		
        \STATE Update $\lambda^{(i)}=\eta \lambda^{(i-1)}$.
		\UNTIL $|\tilde{p}_\mathrm{tot}^{(i)}-\tilde{p}_\mathrm{tot}^{(i-1)}|\leq \epsilon$ and $\mathcal{F}_1(\mathbf{S}_m,\mathbf{U}_k)\leq \epsilon$.
		\STATE \textbf{Output} $\{\phi_{k,n}^m,\tilde{p}_{k,n}^m,\tilde{p}_{n,D}^{m,k},\tilde{p}_{k,D}^m\}$.
	\end{algorithmic}
\end{spacing}
\end{algorithm}
\subsubsection{Quadratic Penalty Approach}
Based on the prior discussions, we propose a QP-based SCA iterative algorithm for solving the problem $\mathrm{P2}$. The algorithm is summarized in Algorithm \ref{Alg2}, which can converge to a locally optimal solution of the problem $\mathrm{P2}$ by solving a sequence of the penalized convex problem $\mathrm{P6\!\!-\!\!QP}$ until convergence \cite{SunIToC2017}.
\begin{algorithm}[htbp]
\caption{\textbf{:} QP-based SCA Iterative Algorithm for Solving Problem $\mathrm{P2}$}
\label{Alg2}
\begin{spacing}{1.3}
	\begin{algorithmic}[1]
		\STATE \textbf{Initialize}~iteration index $i=0$, feasible $(\phi_{k,n}^m)^{(0)},\forall k,m,n\in\mathcal{N}_0$, initial penalty factor $\beta^{(0)}>0$, $\eta>1$, tolerance $\epsilon>0$.
		\REPEAT
        \STATE Set $i=i+1$.
		\STATE Obtain $\{\phi_{k,n}^m,\tilde{p}_{k,n}^m,\tilde{p}_{n,D}^{m,k},\tilde{p}_{k,D}^m\}$ by solving $\mathrm{P6\!\!-\!\!QP}$ and restore $\tilde{p}_\mathrm{tot}^{(i)}$.
        \STATE Update $(\phi_{k,n}^m)^{(i)}=\phi_{k,n}^m,\forall k,m,n\in\mathcal{N}_0$.		
        \STATE Update $\beta^{(i)}=\eta \beta^{(i-1)}$.
		\UNTIL $|\tilde{p}_\mathrm{tot}^{(i)}-\tilde{p}_\mathrm{tot}^{(i-1)}|\leq \epsilon$ and $\mathcal{F}_2(\phi_{k,n}^m)\leq \epsilon$.
		\STATE \textbf{Output} $\{\phi_{k,n}^m,\tilde{p}_{k,n}^m,\tilde{p}_{n,D}^{m,k},\tilde{p}_{k,D}^m\}$.
	\end{algorithmic}
\end{spacing}
\end{algorithm}
\section{Numerical Results}
In this section, some numerical results are presented to evaluate the performance of the proposed relay-aided two-phase transmission protocol. In order to explain the subsequent simulation results explicitly, we plot two examples of the smart factory scenario setting. Then, the performance of the applied two penalty methods, i.e., NCP and QP, are compared in terms of power consumption under each channel realization and the average convergence performance. Moreover, the impacts of the number of robots, the maximum allowable packet error probability, the target number of transmission data bits and the number and locations of relays on the required total transmit power are discussed.

The considered simulation scenario is shown in Fig. \ref{Examples}, where the controller is located in the center of the smart factory which is modeled as a circle of radius $r=300$ meters. The relays are uniformly deployed on the inner circle at a distance $\theta r$ from the controller, where $\theta$ is a distance factor. The total $K$ robots are randomly and uniformly distributed in the smart factory. Note that $d_{k}$ in meter is the distance from the robot $k$ to the controller, $d_{kn}$ denotes the distance from the robot $k$ to the relay $n$, and $d_{n}$ is the distance from the relay $n$ to the controller. The Rayleigh fading components $\tilde{h}_{k,n}^m,\tilde{h}_{n,D}^{m,k},\tilde{h}_{k,D}^m$ are assumed to be independent and identically distributed circularly symmetric complex Gaussian random variables with zero mean and unit variance. The path loss component is modeled as $35.3+37.6\log_{10}(d)$ in dB. Here, two example settings for $K=8,N=2$ and $K=8,N=4$ are shown in Fig. \ref{N2} and Fig. \ref{N4}, which can assist the understanding of some specific simulation results in the subsequent figures. The simulation results are averaged on $100$ channel realizations. Other parameters are listed in Table \ref{table1}, unless otherwise specified.
\begin{figure}[!t]\centering
\subfloat[$K=8,N=2$]{
\label{N2}
\includegraphics[width=.75\linewidth]{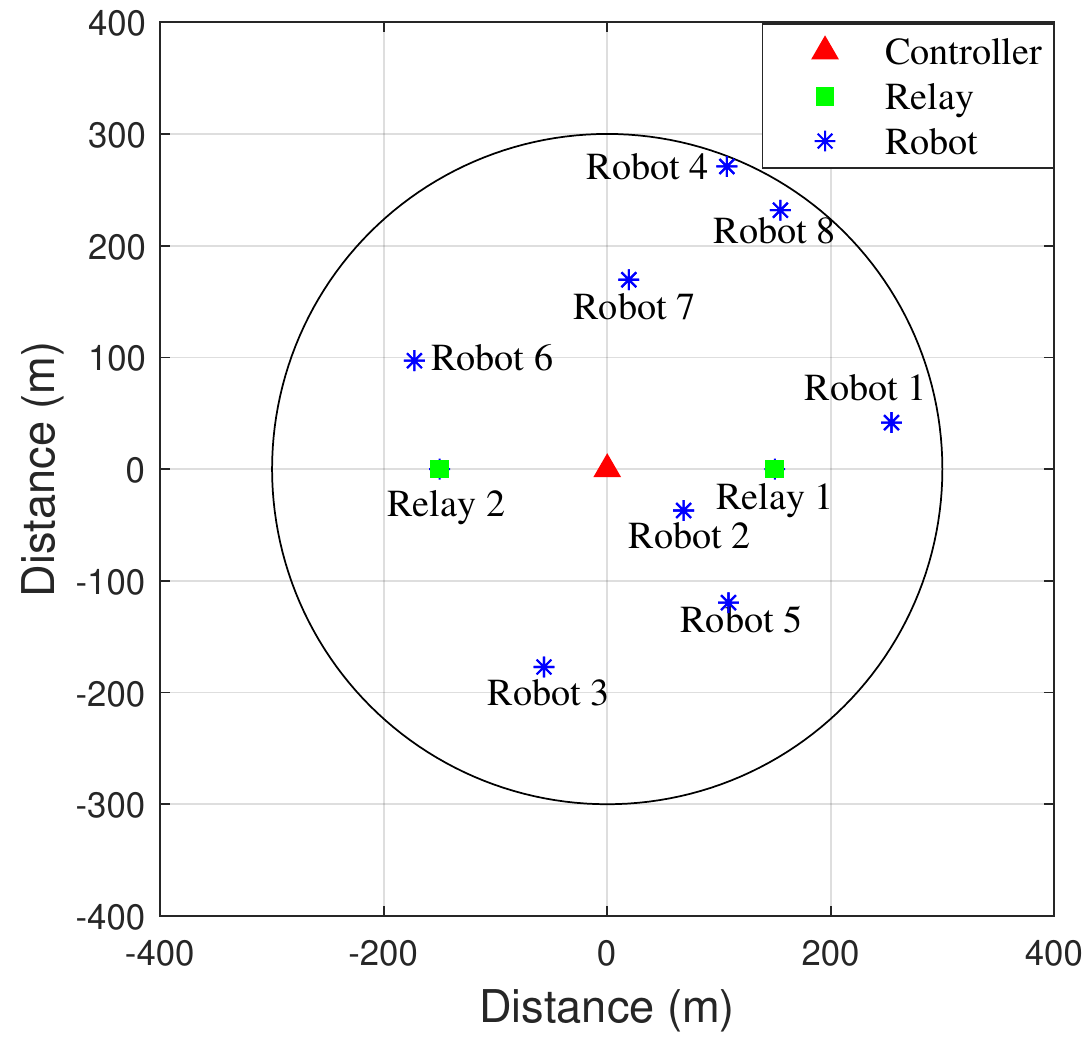}}\\[0.01mm]
\subfloat[$K=8,N=4$]{
\label{N4}
\includegraphics[width=.75\linewidth]{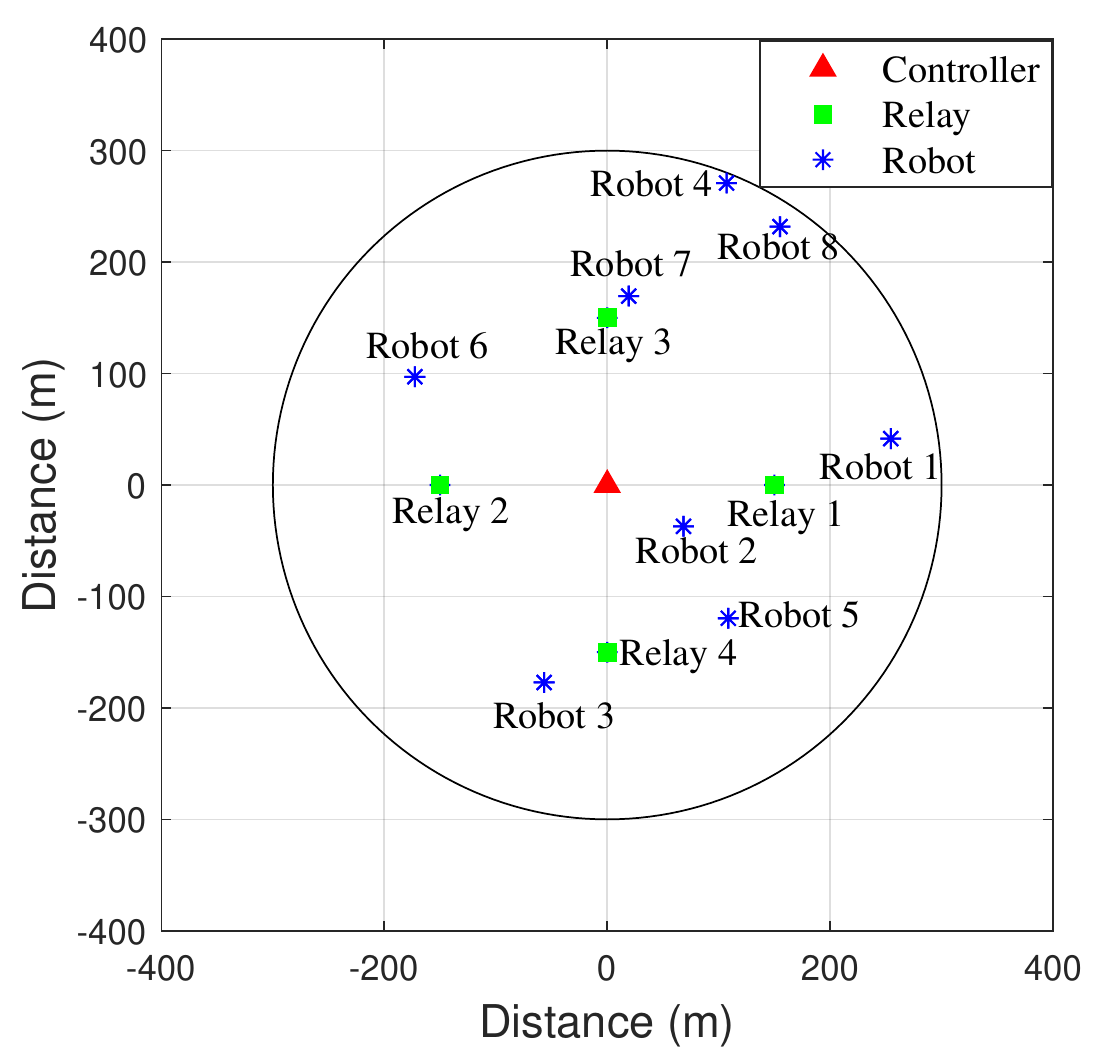}}
\caption{Examples of the smart factory scenario setting.}
\label{Examples}
\end{figure}

\begin{table}[htbp]{
\centering
\caption{Simulation Parameters}
\label{table1}
\setlength{\tabcolsep}{1.2mm}{
\begin{tabular}{lll}
  \toprule
  Symbol & Parameter & Value\\
  \midrule
  $K$ & number of robots & $2,3,4,5,6,7,8$\\
  $M$ & number of RBs & $10$\\
  $N$ & number of relays & $2,4$\\
  $\tau_1$ & duration of Phase I & $0.5$ ms\\
  $\tau_2$ & duration of Phase II & $0.5$ ms\\
  $W$ & bandwidth of each RB & $360$ KHz\\
  $\varepsilon_{\max}$ & maximum PEP & $10^{-9}\!\sim¡¢£¡10^{-3}$\\
  $B_k\!=\!B,\forall k$ & number of data bits & $800,1000$\\
  $N_0$ & noise power spectral density & $-174$ dBm/Hz\\
  $\theta$ & distance factor & $0.2\!\sim\!0.7$\\
  $r$ & radius of smart factory & $300$ m\\
  $\lambda^{(0)}$ & initial penalty factor of NCP & $0.001$\\
  $\beta^{(0)}$ & initial penalty factor of QP & $0.001$\\
  $\eta$ & scaling factor & $2.5$\\
  $\epsilon$ & tolerance & $10^{-4}$\\
  \bottomrule
\end{tabular}}}
\end{table}

Fig. \ref{PerformanceContrast} compares the performance of the NCP-based algorithm and QP-based algorithm from the point views of total transmit power under each channel realization and the average convergence performance. From Fig. \ref{power_vs_channel}, the NCP approach almost performs the same as the QP method under $100$ channel realizations. Meanwhile, from the Fig. \ref{convergence}, it can be readily obtained that both NCP and QP methods can reach fast convergence. This verifies the efficiency of such two penalty approaches for solving the formulated problem. What's more, the QP-based algorithm performs faster in terms of convergence rate.
\begin{figure}[!t]\centering
	\subfloat[$\tilde{p}_\mathrm{tot}$ under $100$ channel realizations]{
    \label{power_vs_channel}
    \includegraphics[width=.8\linewidth]{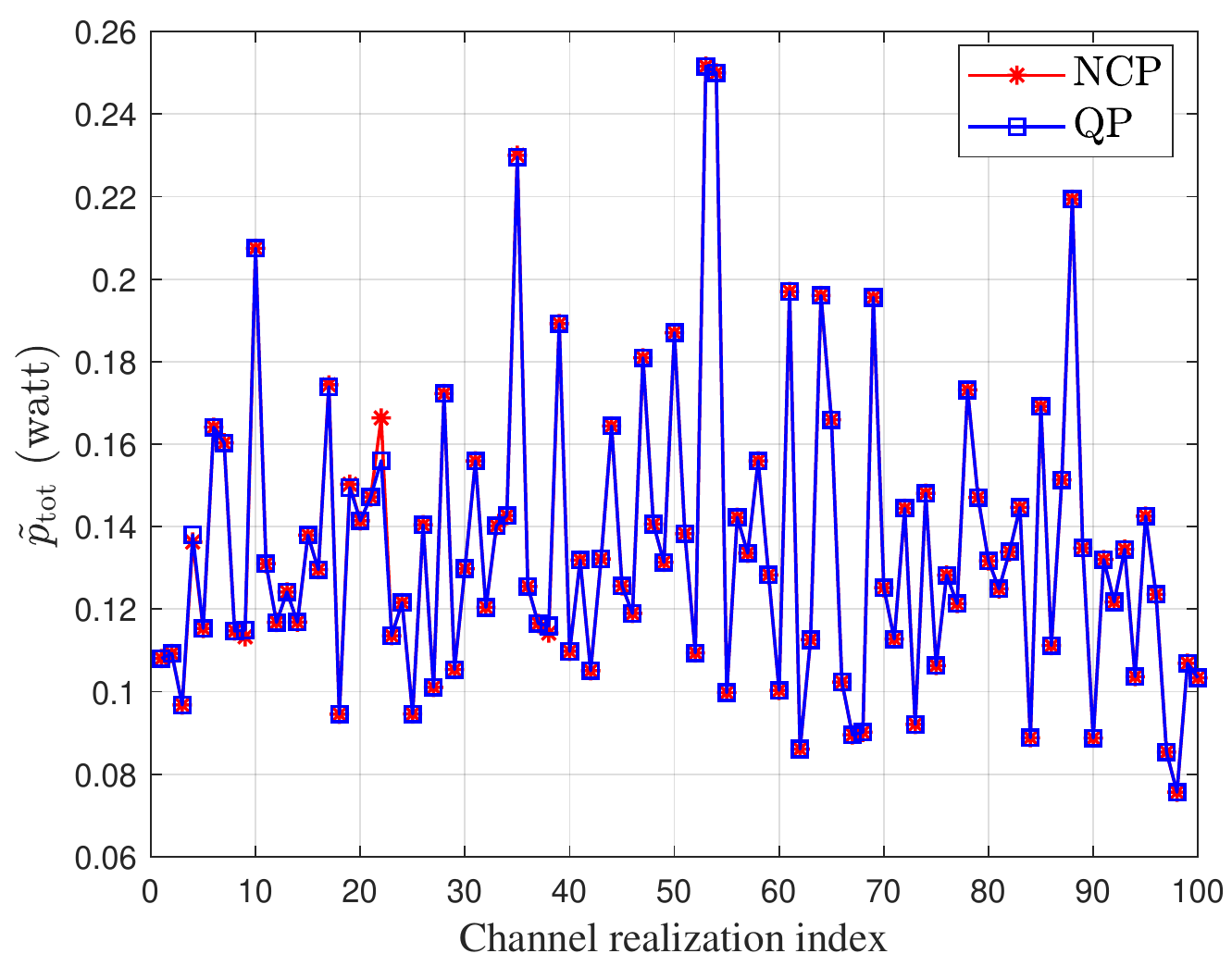}}\\[0.01mm]
	\subfloat[Average convergence performance]{
    \label{convergence}
    \includegraphics[width=.8\linewidth]{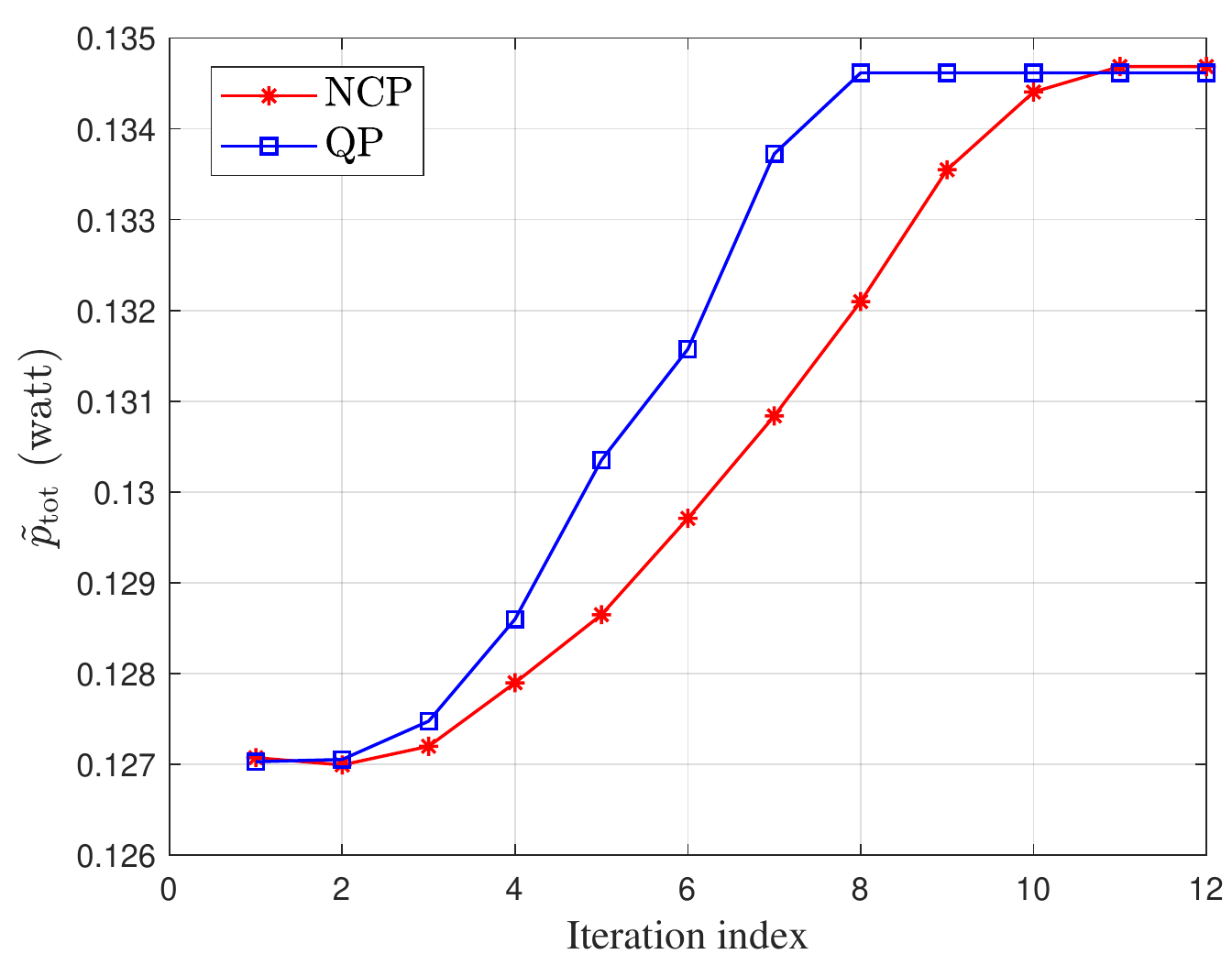}}
	\caption{Performance comparison between NCP and QP for $K=4,N=4,\varepsilon_{\max}=10^{-5},\theta=0.5,B=1000$ bits.}
	\label{PerformanceContrast}
\end{figure}

In Fig. \ref{Ptot_numberofRobots}, the total power consumption significantly increases as the number of robots grows. However, the trend of power consumption increase differs under the different number of relays and reliability requirements. It can be observed that the consumed power is maximum when the number of relays is small and the reliability requirement is extremely high (e.g., $N=2,\varepsilon_{\max}=10^{-8}$). This is because more relays can offer more spatial diversity gain, and a high tolerance for errors (large allowable packet error probability) leads to lower power for achieving the same target rate which can be obtained from the approximated capacity in the finite blocklength regime. Moreover, the impact of the reliability requirement on the power consumption associated with $N=4$ is lower than that associated with $N=2$, which verifies that more spatial diversity gain provided by relays can greatly improve the robots' transmission reliability. Furthermore, it is also seen that when the number of robots is $4$, the total power consumption experiences considerable growth especially for $N=2$ compared with $N=4$. The reason is that the $4$th robot is an edge-user (from Fig. \ref{Examples}), which requires enormous transmit power to enable the ultra-reliable direct transmission of the  target number of data bits. However, for the case of $N=4$, the $4$th robot can transmit packets via the help of the relay $3$, which offers spatial diversity for reliable transmission, thereby lowering the power consumption. In addition, when the number of robots is small (e.g. less than $8$ robots), the corresponding power consumption of the case labeled with $N=2,\varepsilon_{\max}=10^{-5}$ is larger than that of the green dashed line labeled with $N=4,\varepsilon_{\max}=10^{-8}$, but as the number of robots increase (e.g. $8$ robots) especially when the $8$th robot is an edge-user, the situation reverses since the reliability requirement denominates in this case.
\begin{figure}[!t]\centering
	\includegraphics[width=.82\linewidth]{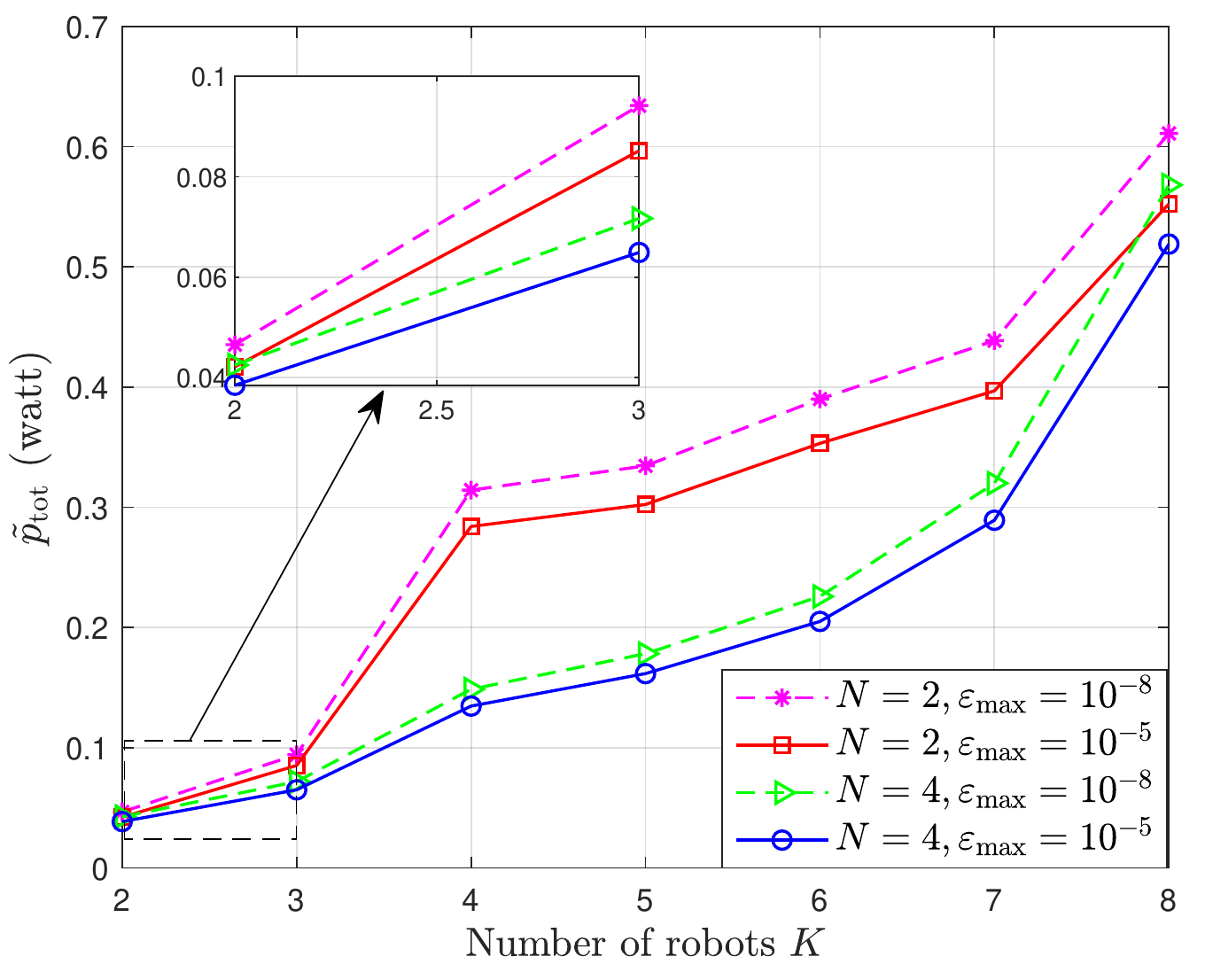}
	\caption{$\tilde{P}_{\mathrm{tot}}$ versus number of robots $K$ under different values of $N$ and $\varepsilon_{\max}$ for $\theta=0.5,B=1000$ bits.}
	\label{Ptot_numberofRobots}
\end{figure}

Fig. \ref{Ptot_epsilon} presents the impact of the reliability requirement on $\tilde{P}_{\mathrm{tot}}$ under different values of $N$ and $B$. As the reliability requirement imposed on the system becomes less demanding, the required total transmit power $\tilde{P}_{\mathrm{tot}}$ decreases significantly. The power gap is considerable especially when the target number of data bits $B$ is large and the available number of relays is small (e.g. $N=2,B=1000$). Moreover, it can be seen that with the use of more relays, the required power for all robots to complete the transmission task specified with $B$ bits and $\varepsilon_{\max}$ can be hugely reduced. This reveals the advantage of the proposed relay-aided two-phase transmission protocol. More specifically, with more relays, there comes more available paths for robots to upload packets to the controller instead of direct transmission, thereby consuming lower power via channels with relatively good channel conditions.
\begin{figure}[!t]\centering
	\includegraphics[width=.83\linewidth]{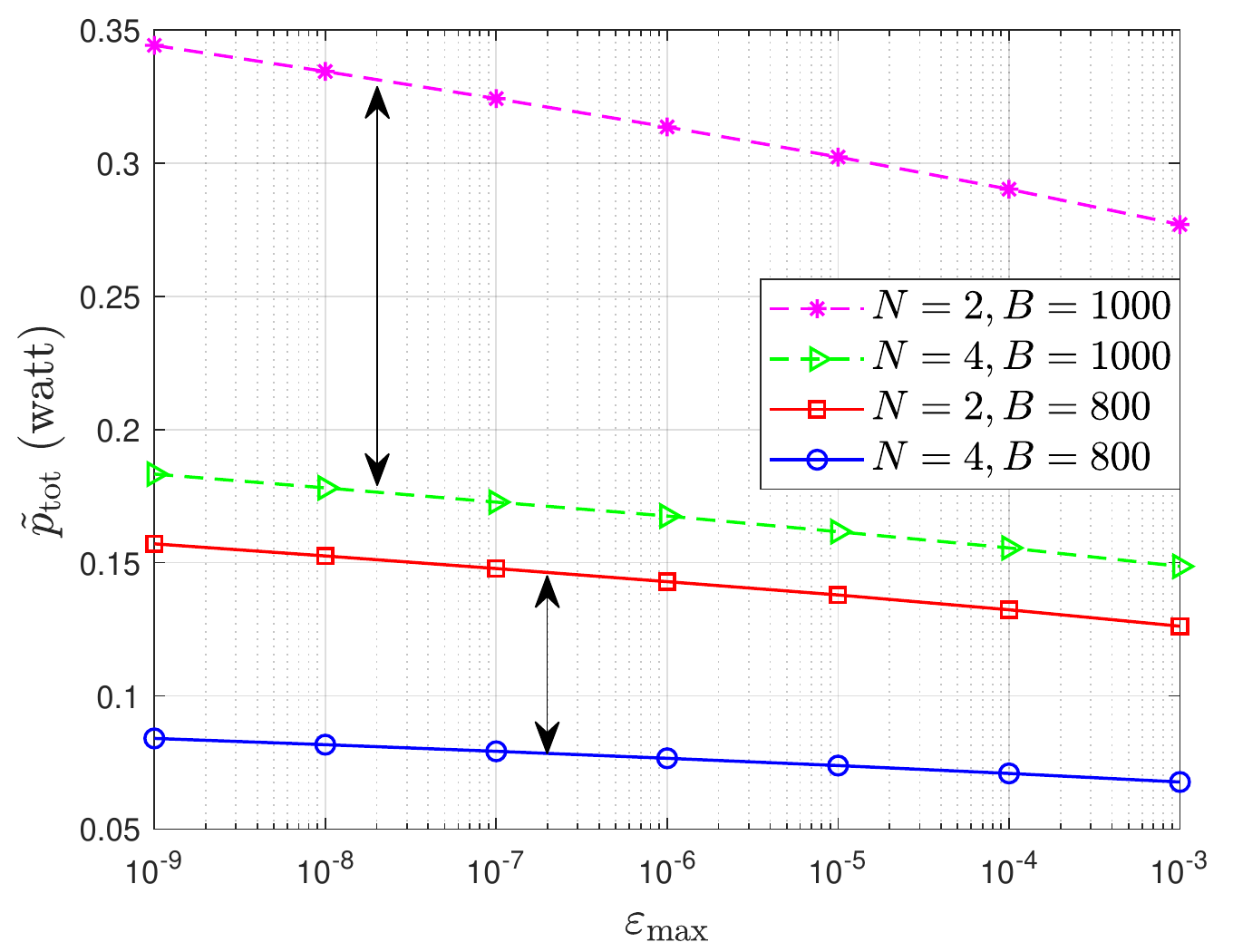}
	\caption{$\tilde{P}_{\mathrm{tot}}$ versus $\varepsilon_{\max}$ under different values of $N$ and $B$ for $\theta=0.5,K=5$.}
	\label{Ptot_epsilon}
\end{figure}

In Fig. \ref{Ptot_theta}, we plot the impact of the relay locations, where $\theta$, a distance factor, is used to qualitatively represent the distance between the relay and the controller. A larger value of $\theta$ indicates a larger distance from the relay to the controller. As can be seen from Fig. \ref{Ptot_theta}, when the relays are deployed nearer or farther from the controller (e.g. $\theta=0.2,0.7$), the required power consumption is high. The reasons are as follows. For the short-distance deployment case (e.g. $\theta=0.2$), the channel condition from the robot to relay is comparative to that from the robot to the controller, thus relays do not function well. Likewise, for the long-distance deployment one (e.g. $\theta=0.7$), the channel condition from the robot to the controller is comparative to that from the relay to the controller, hence direct transmission is preferred. As anticipated, the power consumption $\tilde{P}_{\mathrm{tot}}$ decreases first and then increases with $\theta$. That is, there exists an optimal relay location for deployment. It can be obtained that when $\theta=0.4$, the spatial diversity gain provided by relays can be fully exploited.
\begin{figure}[!t]\centering
	\includegraphics[width=.83\linewidth]{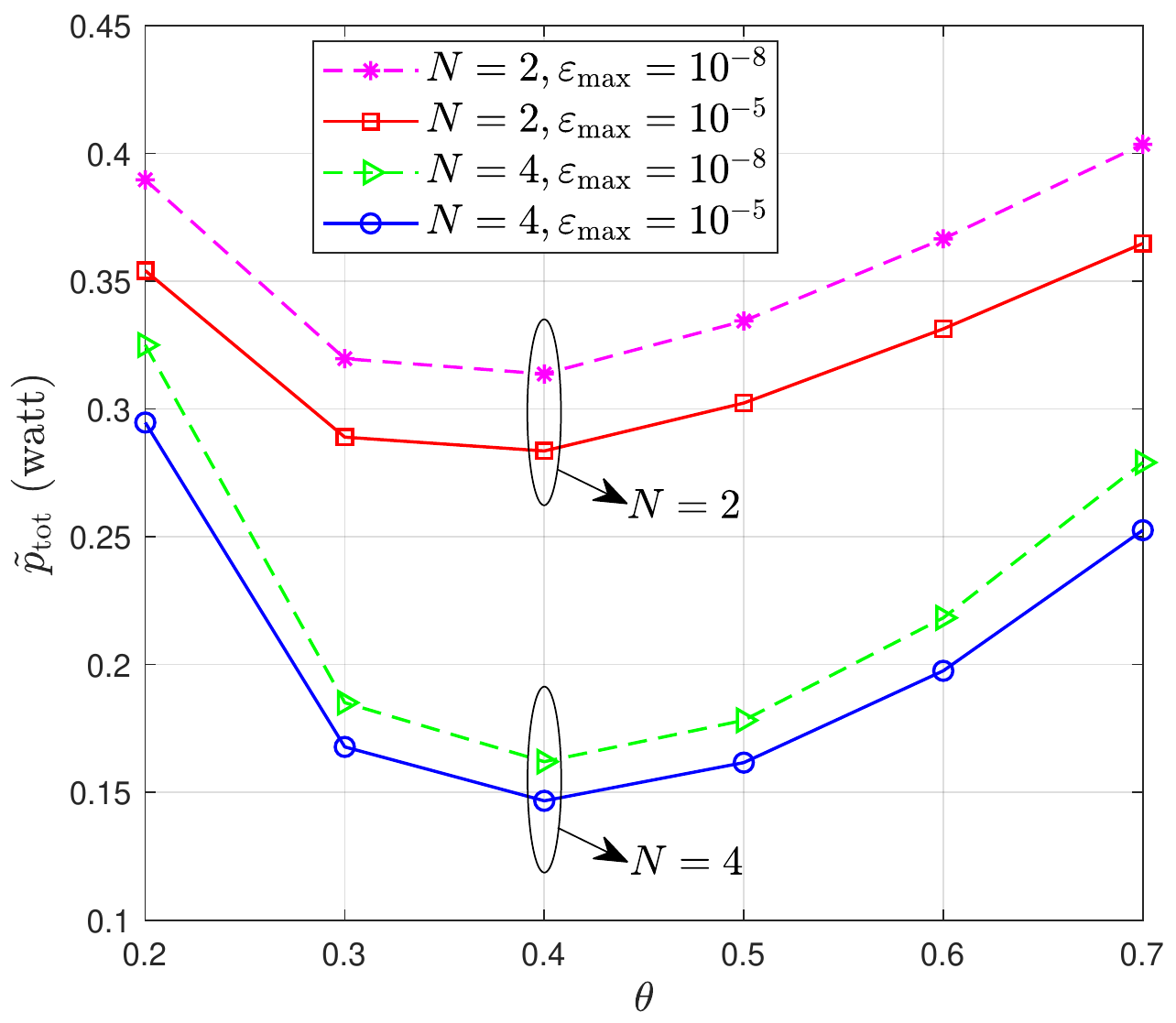}
	\caption{$\tilde{P}_{\mathrm{tot}}$ versus $\theta$ under different values of $N$ and $\varepsilon_{\max}$ for $\theta=0.5,K=5,B=1000$ bits.}
	\label{Ptot_theta}
\end{figure}

In order to further illustrate the influence of relay locations, we plot the proportion of users under different transmission mode marked with direct transmission (DT), assisted transmission by relay $1$, relay $2$, relay $3$ and relay $4$ in Fig. \ref{userPercentage}, respectively, which verifies the discussions about Fig. \ref{Ptot_theta}. More specifically, when $\theta=0.7$, about $60\%$ of robots transmit packets directly to the controller, while up to $75\%$ of robots deliver packets via two-phase relay-aided transmission for the case of $\theta=0.4$. Further, when $\theta=0.4$, the proportion of robots under each transmission mode (if used) is comparative, which demonstrates the optimality of relay locations. Moreover, the utilization rate of the relay $2$ is quite low because most robots are far away from relay $2$ in the example settings.
\begin{figure}[!t]\centering
	\includegraphics[width=.84\linewidth]{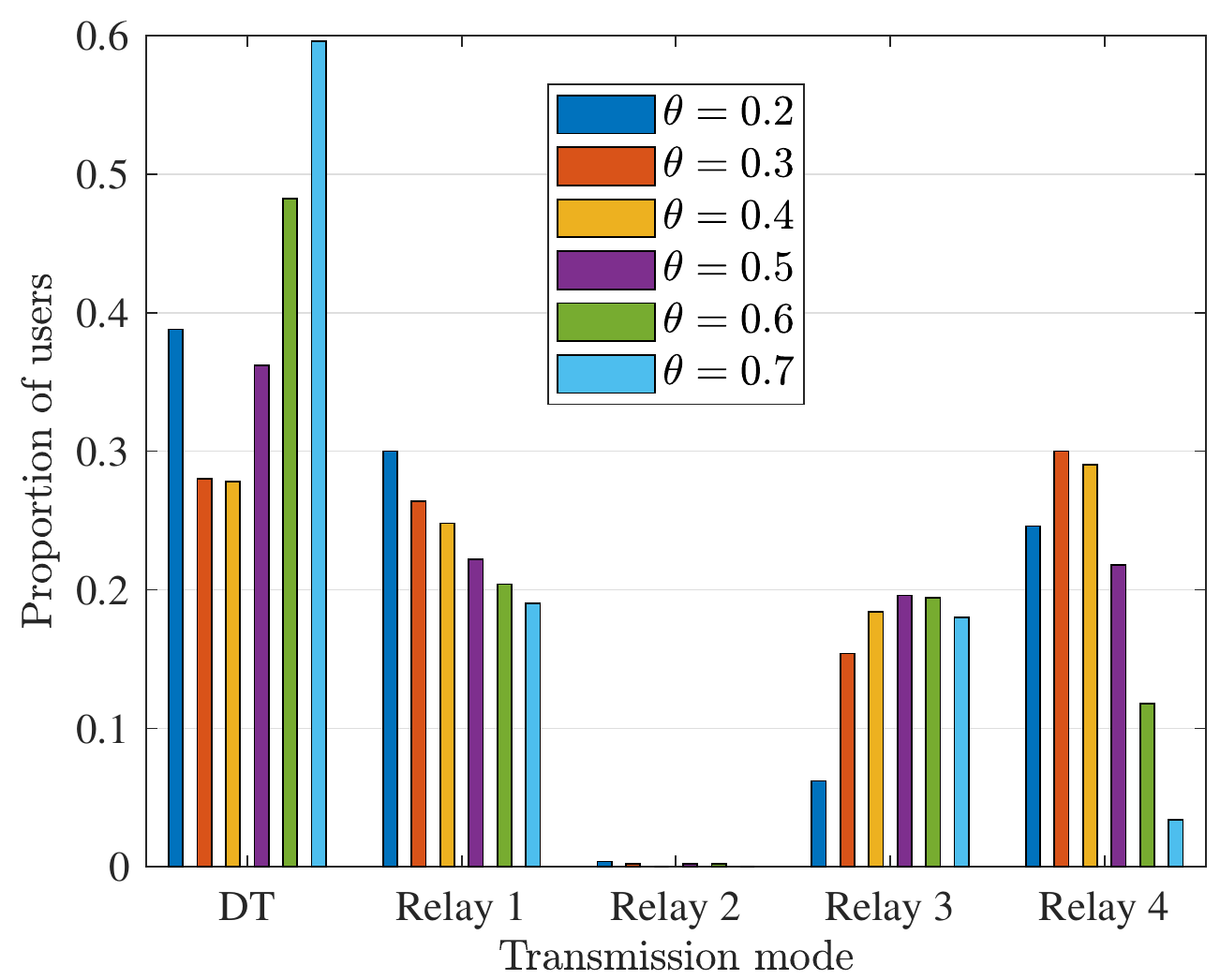}
	\caption{Proportion of users under different transmission mode for $N=4,\varepsilon_{\max}=10^{-5}, \theta=0.5,K=5,B=1000$ bits.}
	\label{userPercentage}
\end{figure}
\section{Conclusion}
In this paper, we have proposed a relay-aided two-phase transmission protocol for the smart factory scenario to enable all robots' ultra-reliable uplink target number of critical data transmission within a predefined delay in an energy-efficient manner. This protocol was formulated as the joint optimization problem of relay selection, RB assignment, and power allocation subject to reliability and target payload constraints for total power minimization. The key challenges towards this formulated problem mainly included neither convex nor concave rate expressions, mutual coupling of the optimization variables, and binary indicator constraints. Different from common methods tailored to tackle the coupling of variables, we fully tapped and utilized the properties of the relative entropy function and obtained the equivalent form with the advantage of not introducing extra constraints. Moreover, in view that the reliability requirements under two possible transmission modes (direct transmission and cooperative transmission) are coupled in one constraint, we first exploited the big-M technique to decouple the reliability requirements in each transmission mode and then adopted a near-optimal combination of decoding error probabilities for the cooperative transmission. Furthermore, NCP and QP approaches have been applied to handle the binary indicator constraints to promote the solutions sparse. Numerical results validate the efficiency of the proposed NCP and QP based algorithms. In addition, the influences of reliability, the number and location of relays, the number of robots, the target number of data bits on the total power consumption are studied, which provide some interesting insights into the future relay-assisted URLLC research.

\begin{appendices}
  \section{Proof of Theorem \ref{theorem}}\label{appendixA}
  The proof can be given based on the abstract Lagrangian duality \cite{Goh2002} theory. Specifically, the Lagrangian associated with the problem $\mathrm{P4\!\!-\!\!QP}$ is defined as
  \begin{equation}
    \mathcal{L}(\Phi,\mathcal{\tilde{P},E},\beta)\!=\!\tilde{P}_{\mathrm{tot}}\!+\!\beta\sum_{k=1}^K \! \sum_{m=1}^M \!\sum_{n=0}^N \!\left(\!\phi_{k,n}^m\!-\!(\phi_{k,n}^m)^2\right).
  \end{equation}
The optimization problem $\mathrm{P4\!\!-\!\!QP}$ can be equivalently rewritten as
\begin{subequations}\label{primalProblem}
\begin{align}
\!\!\!q^{\ast}\!\!=\!\!\min_{\Phi,\mathcal{\tilde{P},E}}\max_{\beta\geq 0} ~&\mathcal{L}(\Phi,\mathcal{\tilde{P},E},\beta)\\
\mathrm{s.t.}~~ &\eqref{P1-C6},\!\eqref{P1-C7},\!\eqref{P3-C1}\!\!-\!\!\eqref{P3-C5},\!
\eqref{two-topReliability},\!\eqref{directReliability},\!\eqref{P4-QP-C2},
\end{align}
\end{subequations}
where $q^{\ast}$ is the optimal value of the problem \eqref{primalProblem}. Meanwhile, the dual problem of \eqref{primalProblem} can be represented as
\begin{subequations}\label{dualProblem}
\begin{align}
\!\!\!d^{\ast}\!=\!\max_{\beta\geq 0}\min_{\Phi,\mathcal{\tilde{P},E}} &\mathcal{L}(\Phi,\mathcal{\tilde{P},E},\beta)\\
\mathrm{s.t.}~~ &\eqref{P1-C6},\!\eqref{P1-C7},\!\eqref{P3-C1}\!\!-\!\!\eqref{P3-C5},\!
\eqref{two-topReliability},\!\eqref{directReliability},\!\eqref{P4-QP-C2},
\end{align}
\end{subequations}
where $d^{\ast}$ is the optimal value of \eqref{dualProblem}. Based on the weak duality, we can have $d^{\ast}\leq q^{\ast}$, i.e.,
\begin{align}\label{maxmin}
  \max_{\beta\geq 0} ~\Omega(\beta)&\triangleq\max_{\beta\geq 0}\min_{\Phi,\mathcal{\tilde{P},E}} ~\mathcal{L}(\Phi,\mathcal{\tilde{P},E},\beta) \notag\\
  &\leq \min_{\Phi,\mathcal{\tilde{P},E}}\max_{\beta\geq 0} ~\mathcal{L}(\Phi,\mathcal{\tilde{P},E},\beta).
\end{align}
Note that $\mathcal{L}(\Phi,\mathcal{\tilde{P},E},\beta)$ is a monotonically increasing function in $\beta$ since $\sum_{k=1}^K \! \sum_{m=1}^M \!\sum_{n=0}^N \!\left(\!\phi_{k,n}^m\!-\!(\phi_{k,n}^m)^2\right)\geq 0$ over the interval $\phi_{k,n}^m\in[0,1],\forall k,m,n\in\mathcal{N}_0$. Thus, $\Omega(\beta)$ is increasing in $\beta$ and is bounded above by $q^{\ast}$. Denote the optimal solutions of the dual problem \eqref{dualProblem} by $\Phi^{\ast},\tilde{\mathcal{P}}^{\ast},\mathcal{E}^{\ast},\beta^{\ast}$. Next, we analyze the solution structure of the dual problem from the following two cases.

Case 1: the optimal solution $\Phi^{\ast}$ of the dual problem \eqref{dualProblem} satisfies $\sum_{k=1}^K \! \sum_{m=1}^M \!\sum_{n=0}^N \!\left(\!\phi_{k,n}^{m\ast}\!-\!(\phi_{k,n}^{m\ast})^2\right)=0$. In this case, $\tilde{\mathcal{P}}^{\ast},\mathcal{E}^{\ast},\beta^{\ast}$ are also feasible solutions of the primal problem \eqref{primalProblem}, thereby obtaining
\begin{equation}\label{case1}
  \tilde{P}_{\mathrm{tot}}(\tilde{\mathcal{P}}^{\ast})=\mathcal{L}(\Phi^{\ast},\tilde{\mathcal{P}}^{\ast},\mathcal{E}^{\ast},\beta^{\ast})=\Omega(\beta^{\ast})\geq q^{\ast}.
\end{equation}
Combing \eqref{maxmin} and \eqref{case1}, we can have the zero duality gap, i.e.,
\begin{equation}
  \max_{\beta\geq 0}\min_{\Phi,\mathcal{\tilde{P},E}} ~\mathcal{L}(\Phi,\mathcal{\tilde{P},E},\beta)=
  \min_{\Phi,\mathcal{\tilde{P},E}}\max_{\beta\geq 0} ~\mathcal{L}(\Phi,\mathcal{\tilde{P},E},\beta)
\end{equation}
must hold for $\sum_{k=1}^K \! \sum_{m=1}^M \!\sum_{n=0}^N \!\left(\!\phi_{k,n}^{m\ast}\!-\!(\phi_{k,n}^{m\ast})^2\right)=0$. Moreover, the monotonicity of $\Omega(\beta)$ w.r.t $\beta$ implies that
\begin{equation}
  \Omega(\beta)=q^{\ast},\forall \beta\geq \beta^{\ast}.
\end{equation}

Case 2: $\sum_{k=1}^K \! \sum_{m=1}^M \!\sum_{n=0}^N \!\left(\!\phi_{k,n}^{m\ast}\!-\!(\phi_{k,n}^{m\ast})^2\right)>0$ holds at the optimal solution of the dual problem \eqref{dualProblem}. In this case, $\max\limits_{\beta\geq 0} ~\Omega(\beta)$ is unbounded since $\Omega(\beta)$ is increasing in $\beta$. This contradicts the inequality in \eqref{maxmin} as the optimal value $q^{\ast}$ is finite and positive.

Therefore, $\sum_{k=1}^K \! \sum_{m=1}^M \!\sum_{n=0}^N \!\left(\!\phi_{k,n}^{m\ast}\!-\!(\phi_{k,n}^{m\ast})^2\right)=0$ holds for the optimal solutions and the proof can be completed.
\end{appendices}

{\small
	\bibliographystyle{IEEEtran}
	\bibliography{Reference}}


%




\ifCLASSOPTIONcaptionsoff
  \newpage
\fi

\end{document}